\input{\jobname.cfg}
\newcommand{\acmart}{\True}

\RequirePackage{boolean}

\documentclass[acmsmall,screen\Extended{,nonacm}{}]{acmart}

\usepackage{boolean}
\usepackage{reversion}

\usepackage[utf8]{inputenc}
\acmart{}{
\usepackage{libertine}
\usepackage{inconsolata}
}

\usepackage{tikz}
\usetikzlibrary{arrows}

\usepackage{mylistings}
\usepackage{mycomments}
\usepackage{mymath}
\usepackage{mybiblio}
\usepackage{myhyperref}
\usepackage{notations}
\usepackage{myrefs}

\UNXXX{}

\begin{document}

\newcommand{\PublishedDOI}{10.1145/3632893}
\newcommand{\arxivId}{2311.07369} 

\Extended{
  \setcopyright{cc}
  \setcctype{by-sa}
}{
  \setcopyright{rightsretained}
}
\acmDOI{\PublishedDOI}
\acmYear{2024}
\copyrightyear{2024}
\acmSubmissionID{popl24main-p252-p}
\acmJournal{PACMPL}
\acmVolume{8}
\acmNumber{POPL}
\acmArticle{51}
\acmMonth{1}
\received{2023-07-11}
\received[accepted]{2023-11-07}

\title[Unboxed Data Constructors]{Unboxed Data Constructors: \\ Or, How \texttt{cpp} Decides a Halting Problem}

\author{Nicolas Chataing}
\orcid{0009-0006-4174-2088}
\affiliation{%
  \institution{ENS Paris}
  \country{France}
}
\email{nicolas.chataing@gmail.com}

\author{Stephen Dolan}
\orcid{0000-0002-4609-9101}
\affiliation{%
  \institution{Jane Street}
  \country{UK}
}
\email{stedolan@stedolan.net}

\author{Gabriel Scherer}
\orcid{0000-0003-1758-3938}
\affiliation{%
  \institution{Inria}
  \country{France}
}
\email{gabriel.scherer@gmail.com}

\author{Jeremy Yallop}
\orcid{0009-0002-1650-6340}
\affiliation{%
  \institution{University of Cambridge}
  \country{UK}
}
\email{jeremy.yallop@cl.cam.ac.uk}

\begin{abstract}
  We propose a new language feature for ML-family languages, the ability to selectively \emph{unbox} certain data constructors, so that their runtime representation gets compiled away to just the identity on their argument.
  Unboxing must be statically rejected when it could introduce confusion, that is, distinct values with the same representation.
  
  We discuss the use-case of big numbers, where unboxing allows to write code that is both efficient and safe, replacing either a safe but slow version or a fast but unsafe version.
  We explain the static analysis necessary to reject incorrect unboxing requests.
  We present our prototype implementation of this feature for the OCaml programming language, discuss several design choices and the interaction with advanced features such as Guarded Algebraic Datatypes.

  Our static analysis requires expanding type definitions in type expressions, which is not necessarily normalizing in presence of recursive type definitions.
  In other words, we must decide normalization of terms in the first-order $\lambda$-calculus with recursion.
  We provide an algorithm to detect non-termination on-the-fly during reduction, with proofs of correctness and completeness.
  Our algorithm turns out to be closely related to the normalization strategy for macro expansion in the \texttt{cpp} preprocessor.
\end{abstract}



\newcommand{\PublishedURL}{https://doi.org/\PublishedDOI}
\newcommand{\ExtendedWebsiteURL}{https://arxiv.org/abs/\arxivId}
\newcommand{\ExtendedDocURL}{https://arxiv.org/pdf/\arxivId}

\begin{version}{\Extended}
\titlenote{
    A version of this document was published at POPL'24
    (\url{\PublishedURL}) without appendices.
    The present extended version includes appendices, and is archived
    at \url{\ExtendedWebsiteURL}.
}
\end{version}
\begin{version}{\BNot\Extended}
\titlenote{
    The extended version of this work, with appendices,
    is available at \url{\ExtendedWebsiteURL}.
}
\end{version}

\maketitle

\newcommand{\appendixref}[1]{%
  \Extended{Appendix~\ref{#1}}{%
    an appendix in our \href{\ExtendedDocURL}{extended version}%
  }%
}

\section{Introduction}

\subsection{Sum Types and Constructor Unboxing}

A central construct of ML-family programming languages is
\emph{algebraic datatypes}, in particular \emph{(disjoint) sum types},
also called \emph{variant types}. In OCaml syntax:

\begin{minipage}{0.5\linewidth}
\begin{lstlisting}
type rel_num =
| Pos of nat
| Zero
| Neg of nat
\end{lstlisting}
\end{minipage}
\hfill
\begin{minipage}{0.5\linewidth}
\begin{lstlisting}
let non_negative = function
| Pos(_) -> true
| Zero -> true
| Neg(_) -> false
\end{lstlisting}
\end{minipage}

Values of this sum type are of the form \lstinline{Pos(n)} or
\lstinline{Zero} or \lstinline{Neg(n)}; in particular they start with
a \emph{(data) constructor}, here \lstinline{Pos} or \lstinline{Zero}
or \lstinline{Neg}, followed by \emph{arguments} (zero, one or
several arguments), which is data carried along the constructor. In
mathematical terms this corresponds to a \emph{sum} or
\emph{coproduct} $A + B$ between sets of values, rather than a
\emph{union} $A \cup B$, because one can always tell from which side
of the sum a value is coming from, by pattern-matching.

In mathematics, the coproduct between sets is typically implemented as
a (disjoint) union:
\begin{mathline}
  A + B
  \ \eqdef\ %
  (\{0\} \times A) \uplus (\{1\} \times B)

  \sum\Fam {i \in I} {A_i}
  \ \eqdef\ %
  \biguplus_{i \in I}\, (\{i\} \times A_i)
\end{mathline}
using a cartesian product of the form $\{i\} \times S$ to build
the pairs of a \emph{tag} value $i$ and an element of $S$.

Software implementations of programming languages use the same
approach: the representation of sum types in memory typically includes
not only the data for their arguments, but also a \emph{tag}
representing the constructor; pattern-matching on the value is
implemented by a test on this tag -- often followed by accessing the
arguments of the constructor. For example, the standard OCaml compiler
will represent the value \lstinline{Neg(n)} by a pointer to a block of
two consecutive machine words in heap memory, with the first word
(the block \emph{header}) containing the tag (among other things),
followed by the argument of type \lstinline{nat} in the second
word. We say that the parameter of type \lstinline{nat} is
\emph{boxed} in this representation: it is contained inside another
value, a ``box''.

Boxing induced by data constructors introduces a performance
overhead. In the vast majority of cases it is negligible, as ML-family
language implementations heavily optimize the allocation of heap
blocks and often benefit from good spatial locality. There are still
some performance-critical situations where the overhead of boxing
(compared to carrying just a \lstinline{nat}) is significant.

One common, easy case where boxing can be avoided is for sum types
with a single constructor: \lstinline{type t = Foo of arg}. Values of
this type could be represented exactly as values of type
\lstinline{arg}, as there are no other cases to distinguish than
\lstinline{Foo}. Haskell provides an explicit \lstinline{newtype}
binder for this case. In OCaml, this is expressed by using the
\lstinline{@@unboxed} attribute on the datatype declaration.\\
\begin{minipage}{0.45\linewidth}
\begin{lstlisting}[language=Haskell]
newtype Fd = Fd Int  -- Haskell
\end{lstlisting}
\end{minipage}
\begin{minipage}{0.55\linewidth}
\begin{lstlisting}
type fd = Fd of int [@@unboxed] (* OCaml *)
\end{lstlisting}
\end{minipage}

Note that the programmer could have manipulated values of type
\lstinline{int} directly, instead of defining a single-constructor
type \lstinline{fd} (file descriptor) isomorphic to \lstinline{int};
but often the intent is precisely to define an isomorphic but
incompatible type, to have explicit conversions back and forth in the
program, and avoid mistaking one for the other. Single-constructor
unboxing makes efficient programming more expressive or safe, or
expressive/safe programming more efficient.

\subsection{Constructor Unboxing}

In the present work, we introduce a generalization of
single-constructor unboxing, that we simply call \emph{constructor
  unboxing}. It enables unboxing one or several
constructors of a sum type, as long as disjointedness is
preserved.

Our main example for this selective unboxing of constructors is
a datatype \lstinline{zarith} of (relative) numbers of arbitrary size,
that are represented either by an OCaml machine-word integer
(\lstinline{int}), when they are small enough, or a ``big number'' of
type \lstinline{Gmp.t} implemented by the GMP library~\citep*{gmp}:
\begin{lstlisting}
type zarith = Small of int [@unboxed] | Big of Gmp.t
\end{lstlisting}

The \lstinline{[@unboxed]} attribute requests the \emph{unboxing} of
the data constructor \lstinline{Small}, that is, that its application
be represented by just the identity function at runtime. This request
can only be satisfied if it does not introduce \emph{confusion} in
the datatype representation, that is, two distinct values (at the
source level) that get the same representation. In our \code{zarith}
example, the definition can be accepted with an unboxed constructor,
because the OCaml representation of \lstinline{int} (immediate values)
is always disjoint from the boxed constructor \lstinline{Big of Gmp.t}
(heap blocks). Otherwise we would have rejected this definition
statically.

\begin{lstlisting}
type clash = Int of int [@unboxed] | Also_int of int [@unboxed]
/[Error]/: This declaration is invalid, some [@unboxed] annotations introduce
        overlapping representations.
\end{lstlisting}

While constructor boxing is cheap in general, unboxing the
\code{Small} constructor of the \code{zarith} type does make
a noticeable performance difference, because arithmetic operations in
the fast path are very fast, even with an explicit overflow check, so
the boxing overhead would be important. On a synthetic benchmark,
unboxing the \lstinline{Small} constructor provides a 20\% speedup.

\subsection{Head Shapes}

We propose a simple criterion to statically reject unboxing requests
that would introduce confusion -- several distinct values with the
same runtime representation -- parameterized on two notions:
\begin{itemize}
\item The \emph{head} of a value, an approximation/abstraction of
  its runtime representation.
\item The \emph{head shape} of a type, the (multi)set of heads of
  values of this type.
\end{itemize}

Our static analysis computes the head shape of datatype definitions,
and rejects definitions where the same head occurs several times.

In addition, we require that the head of a value be efficiently
computable at runtime. In presence of an unboxed constructor, the head
shape of its argument type is used to compile pattern-matching. The
generated code may branch at runtime on the head of its scrutinee.

We provide a definition of heads that is specific to the standard
OCaml runtime representation. Other languages would need to use
a different definition, but the static checking algorithm and the
pattern-matching compilation strategy are then fully
language-agnostic.

\subsection{A Halting Problem}

To compute the head shape of types, and therefore to statically check
unboxed constructors for absence of confusion, we need to unfold
datatype abbreviations or definitions. Example:\\
\begin{minipage}{0.5\linewidth}
\begin{lstlisting}
type num = int
and name = Name of string [@unboxed]
\end{lstlisting}
\end{minipage}
\hfill
\begin{minipage}{0.5\linewidth}
\begin{lstlisting}
type id =
  | By_number of num [@unboxed]
  | By_name of name [@unboxed]
\end{lstlisting}
\end{minipage}

To check the final definition of \lstinline{id}, we must determine
that \lstinline{num} is the primitive type \lstinline{int} and that
\lstinline{name} has the same representation as a primitive
\lstinline{string}, through a definition-unfolding process that
corresponds to a form of $\beta$-reduction.

In the general case, type definitions may contain arbitrary
recursion. In particular, unfolding definitions may not terminate for
some definitions; we need to detect this to prevent our static
analysis from looping indefinitely.

\begin{lstlisting}
type 'a id = Id of 'a [@unboxed]
type loop = Loop of loop id [@unboxed]
\end{lstlisting}

This practical problem is in fact exactly the halting problem --
deciding whether terms have a normal form -- for the \emph{pure,
  first-order $\lambda$-calculus with recursion}. We present an algorithm to
detect non-termination on the fly, during normalization. Running the
algorithm on a term either terminates with a normal form, or it stops
after detecting a loop in finite time. We prove that our algorithm is
correct (it rejects all non-terminating programs) and complete
(it accepts all terminating programs). The proof of correctness
requires a sophisticated termination argument.

It turns out that this termination-monitoring algorithm is related to
the approach that the \texttt{cpp} preprocessor uses to avoid
non-termination for recursive macros -- another example where
non-termination must be detected on the fly.

\subsection{Contributions and Outline}

We claim the following contributions:
\begin{enumerate}
\item We specify the feature of ``constructor unboxing''.

\item We implement this feature in an experimental version of the
  OCaml compiler\footnote{We included an anonymized git repository
    with our submission.}, and hope to eventually merge it
  upstream. Section~\ref{sec:scaling} details various design
  considerations that arose when scaling the feature to a full
  programming language.

\item We perform a case study on \texttt{zarith}
  (Section~\ref{sec:zarith}), demonstrating that the feature can
  noticeably improve the performance of safe OCaml code, or noticeably
  improve safety by removing the need for unsafe OCaml code that was
  previously used for this purpose.

\item We provide a static analysis to reject unboxing requests that
  would introduce confusion. Our formal treatment
  (Section~\ref{sec:head-shapes}) is independent from the OCaml
  programming language.

\item As an unplanned side-effect, we propose an \emph{on-the-fly
    termination checking} algorithm for the pure, first-order
  $\lambda$-calculus with recursion
  (Section~\ref{sec:halting-problem}). We prove that this algorithm is
  correct (it rejects all non-terminating programs) and complete
  (it accepts all terminating programs).

\item There are unexpected parallels between our termination checking
  work and an algorithm written in 1986 by Dave
  Prosser~\citep*{prosser-86}, to ensure termination of \code{cpp}
  macro expansions. In Section~\ref{sec:cpp} we compare our algorithm
  to \code{cpp}'s -- for which correctness has, to our knowledge, not
  been proved.
\end{enumerate}

Our Related Work Section~\ref{sec:related-work} has a mix of
discussions on production programming languages, experimental language
designs, and theoretical work on termination of recursive
$\lambda$-calculi. We list unboxing-related features in other
languages and implementations (GHC Haskell, MLton, \Fsharp{},
Scala, Rust). We mention research projects that provide an explicit
description language for data layout (Hobbit, Dargent, Ribbit). Finally, we
discuss Rust niche-filling optimizations in more detail.

\section{Case Study: Big Integers}
\label{sec:zarith}


\subsection{A Primer on OCaml Value Representations}
\label{subsec:ocaml-value-representation}

In the reference implementation of OCaml, sum types have the following
in-memory representation:

\begin{itemize}
\item \emph{constant} constructors (without parameters), such as
  \code{[]} or \code{None}, are represented as \emph{immediate
    values}, exactly like OCaml integers. The immediate values used
  are \code{0} for the first constant constructor in the type,
  \code{1} for the second one, etc.

\item \emph{non-constant} constructors, with one or several
  parameters, are represented by pointers to memory \emph{blocks},
  that start with a header word followed by the representation of each
  parameter. The header word contains some information of interest for
  the GC and runtime system, including a one-byte \emph{tag}: \code{0}
  for the first non-constant constructor, \code{1} for the second one,
  etc., and the \emph{arity} of the constructor -- its number of
  parameters.
\end{itemize}

Some primitive types that are not sum types -- strings, floats, lazy
thunks, etc. -- have special support in the OCaml runtime. They are
also represented as memory blocks with a header word, using a dozen
reserved high tag values that cannot be used by sum constructors. They
also include a tag \texttt{Custom\_tag} (255) for blocks whose
parameters are foreign data words accessed only through the C Foreign
Function Interface (FFI).

OCaml distinguishes immediate values from pointers to blocks
(both word-sized) by reserving the least significant bit for this
purpose: immediate values are encoded in odd words, while pointers are
even words. (In particular, OCaml integers are 63-bits on
64-bit machines.) Pattern-matching can check this
immediate-or-pointer bit first, to distinguish constant from
non-constant constructors, and then switch on a small immediate
value or block tag.

\subsection{Unsafe Zarith}

\newcommand{\Zarith}{\href{https://github.com/ocaml/Zarith}{Zarith}}

The \Zarith{} library~\citep*{zarith} provides an efficient OCaml
implementation of arbitrary-precision integers, on top of the
reference C library \code{Gmp}~\citep*{gmp}.

Some users of arbitrary-precision integers perform a majority of their
computations on very large integers, way larger than the ``small''
integers that fit a machine word. On the other hand, many users
perform computations that rarely, if ever, overflow, but they need the
guarantee that the result will remain correct even in presence of
occasional overflows. For this latter use-case, we want to minimize
the overhead of \Zarith{} compared to using machine-sized integers
directly -- in OCaml, the \code{int} type. We want to ensure that when
operating on ``small'' integers, the operation only performs
machine-size arithmetic and bound checks, without any memory
allocation, nor any call to non-trivial C functions; in other words,
we want a ``fast path'' for small integers.

\Zarith{} uses a type \code{Zarith.t} whose inhabitants are
either machine-sized OCaml integers (type \code{int}) or a ``custom''
value, a pointer to a memory block with tag \texttt{Custom\_tag} and
the \code{gmp} digits as arguments. This type cannot be expressed in
OCaml today, so \Zarith{} has to use the low-level, unsafe compiler
intrinsics to perform unsafe checks and casts, giving up on the type-
and memory-safety usually guaranteed by the OCaml programming
language.\footnote{A previous version of Zarith would even implement
  the small-integer fast path in assembly code on some architectures
  to use native overflow-checking instructions instead of bit-fiddling
  checks in OCaml. But the cost of switching from OCaml to the
  assembly FFI in the fast path in fact made this version slower than
  the OCaml version -- it was also painful to maintain.}

\begin{lstlisting}
type t                     (* int or gmp integer (in a custom block) *)
external is_small_int: t -> bool = "%obj_is_int" (* imm-or-block bit *)
external unsafe_to_int: t -> int = "%identity"        (* unsafe cast *)
external of_int: int -> t = "%identity"               (* unsafe cast *)

external c_add: t -> t -> t = "ml_z_add"          (* slow path, in C *)
let add x y =
  if is_small_int x && is_small_int y then begin (* ``fast path'' addition *)
    let z = unsafe_to_int x + unsafe_to_int y in
    (* Overflow check -- Hacker's Delight, section 2.12 *)
    if (z lxor unsafe_to_int x) land (z lxor unsafe_to_int y) >= 0
    then of_int z else c_add x y
  end else (* ``slow path'' *) c_add x y
\end{lstlisting}

\subsection{Unboxed Zarith}
\label{subsec:unboxed-zarith}

Using our experimental OCaml compiler with constructor unboxing, we can write instead:
\begin{lstlisting}
type custom_gmp_t [@@shape [custom]] (* gmp integer (in custom block) *)
type t = Small of int [@unboxed] | Big of custom_gmp_t [@unboxed]

external c_add: t -> t -> t = "ml_z_add"
let add a b = match a, b with
  | Small x, Small y ->
      let z = x + y in
      (* Overflow check -- Hacker's Delight, section 2.12 *)
      if (z lxor x) land (z lxor y) >= 0
      then Small z else c_add a b
  | _, _ -> c_add a b
\end{lstlisting}

This code is equivalent to the previous version, generates exactly the
same machine code, but does not require any unsafe casts. One still has
to trust the FFI code of \code{c_add} to respect the intended memory
representation, and we trust the annotation \code{shape [custom]} that
claims that the abstract type \code{custom_gmp_t} is only inhabited
(via the FFI) by \texttt{Custom\_tag}-tagged blocks. But the unsafe
boundary has been pushed completely off the fast path; it can
disappear completely in other examples not involving bindings to
C libraries.

On a synthetic microbenchmark, we observed that our new version has
essentially the same performance as the previous unsafe code, and is
20\% more efficient than a boxed version -- using a sum type without
\code{[@unboxed]} annotations.

\paragraph{Case-study conclusion} Unboxing relieves some of the
tension between safety and efficiency in performance-critical
libraries. Users sometimes have to choose between safe, idiomatic sum
types or more efficient encodings that are unsafe and require more
complex code. In some cases, such as \Zarith{}, constructor unboxing
provides a safe, clear and efficient implementation.

\subsection{Other Use-Cases} Let us briefly mention a few other
use-cases for constructor unboxing.

\paragraph{A ropes benchmark} Our original design proposal
\citet*{unboxing-RFC} includes a similar performance experiment on
ropes (trees formed by concatenating many small strings), reporting
a 30\% performance gain on an example workload -- a similar
performance ballpark to our 20\%. The example was implemented using
unsafe features only as unboxing was not implemented at the time; we
can express it as follows:
\begin{lstlisting}
type rope =
  | Leaf of string [@unboxed]
  | Branch of { llen: int; l:rope; r:rope }
\end{lstlisting}

\paragraph{Coq's \code{native_compute} machinery}

Another use-case where constructor unboxing could provide safety is
the representation of Coq values in the \code{native_compute}
implementation of compiled reduction, first introduced in
\citet*{native-compute}. \code{native_compute} is a Coq tactic that
compiles a Coq term into an OCaml term such that evaluating the OCaml
term (by compilation then execution, in the usual call-by-value
OCaml strategy) computes a strong normal form for the original Coq
term. It uses an unsafe representation of values that mixes (unboxed)
functions, sum constructors, and immediate values, and could be
defined as a proper OCaml inductive if constructor unboxing was
available. (The relation with constructor unboxing was
\href{https://github.com/ocaml/RFCs/pull/14#issuecomment-674442570}{pointed
out to us} by Jacques-Henri Jourdan, Guillaume Melquiond and Guillaume
Munch-Maccagnoni.)

\paragraph{A partial sum-type presentation of dynamic values}

This feature could make some forms of dynamic introspection of runtime
values more ergonomic than what is currently exposed in the \code{Obj}
module. We could think of defining a sort of ``universal type'' as
follows:
\begin{lstlisting}
type dyn =
| Immediate of int [@unboxed]
| Block of dyn array [@unboxed]
| Float of float [@unboxed]
| String of string [@unboxed]
| Function of (dyn -> dyn) [@unboxed]
| Custom of custom [@unboxed]
| ...
and custom [@@shape [custom]]

let to_dyn : 'a -> dyn =
  fun x -> (Obj.magic x : dyn)
\end{lstlisting}

This interface cannot cover all needs -- for example it is not
possible to distinguish between \code{Int32}, \code{Int64} by their
tags, they would be lumped together in the \code{Custom} case -- and
we may need to restrict it due to shape approximations required by
portability concerns -- see Section~\ref{subsec:portability}. But it
would still provide a pleasant pattern-matching interface for unsafe
value introspection code that people write today using the \code{Obj}
module directly.\footnote{A Github code search for \code{Obj} usage suggests
\url{https://github.com/rickyvetter/bucklescript/blob/cbc2bd65ce334e1fc83e1c4c5bf1468cfc15e7f9/jscomp/ext/ext_obj.ml\#L25}
for example.}

\paragraph{Non-use-case: magic performance gains in many places} One
should \emph{not} hope that there are plenty of performance-sensitive
OCaml codebases lying around today that would get a noticeable
performance boost by sprinkling a few (or many) unboxing
annotations. In the vast majority of cases, unboxing provides no
noticeable performance improvement. There are two reasons:
\begin{enumerate}
\item Allocating values in the OCaml minor heap is \emph{really
    fast}. In the boxed version of the Zarith benchmark, checking
  against integer overflow is slower than allocating the \code{Small}
  constructor. Most programmers overestimate the performance cost of
  boxing, it makes little difference for most workloads.\footnote{One
    should be careful that the cost of boxing is not only the extra
    allocation, but the reduced memory locality of wider data
    representations. Locality effects are hard to measure accurately
    and in particular are not captured well in
    micro-benchmarks. Still, our prediction remains that unboxing
    makes a very small difference for most use-cases.}

\item In the few cases of performance-sensitive programs where boxing
  would add noticeable overhead, the authors of the program already
  chose a different implementation strategy to avoid this boxing, for
  example using unsafe tricks as Zarith.
\end{enumerate}

The second point is common to most language design for performance:
\emph{existing} performance-sensitive codebases are written with the
existing feature set in mind, and typically do not present low-hanging
fruits for a new performance-oriented feature. The benefits rather
come from giving more, better options (safer, simpler, more idiomatic)
to write performant code in the \emph{future}.

It is also reassuring for users to know that a new idiom made
available is ``zero-cost''. Even in cases where in fact the
non-optimized approach would have perfectly fine performance, there is
a real productivity benefit for users to know that a given change has
zero performance impact. For example, this can avoid the need to write
specific benchmarks to reassure reviewers of a change.

\section{Heads and Head Shapes}
\label{sec:head-shapes}

To formalize constructor unboxing, we use a simple language of types
$\tau$ and datatype definitions $d$. This captures sum types in
typical ML-inspired, typed functional programming languages.

\paragraph{Notation} We write $\Fam {i \in I} {e_i}$ for a family of
objects $e_i$ indexed over $I$. (Placing the domain as a superscript
is reminiscent of the exponential notation $A^I$ for function spaces.)
We often omit the indexing domain $I$, writing just
$\Fam i {e_i}$. Indexing domains $I$, $J$, $K$, etc.~are
meta-variables that we understand as denoting finite, totally ordered
sets -- for example, integer intervals $[0; n]$.

\subsection{Types, Datatype Declarations, Values}
\begin{mathpar}
\begin{array}{rcl}
\mathsf{Types} \ni \tau
& \bnfeq
& \alpha
  \bnfor \tyconstr t {\Fam i {\tau_i}}
  \bnfor \tyconstr {\prim{t}} {\Fam i {\tau_i}}
\\
\mathsf{TyDecls} \ni d
& \bnfeq
& \typedecl
    {\tyconstr t {\Fam i {\alpha_i}}}
    {\Fam j {C_j \of \Fam {k \in K_j} {\tau_{j,k}}}}
    {\Fam l {C^\unboxed_u \of \tau_l}}
\\
\end{array}
\end{mathpar}

A datatype definition introduces a datatype constructor $t$ parameterized
over the family of type variables $\Fam i {\alpha_i}$ as a sum type
made of a (possibly empty) family of boxed constructors
$\Fam j {C_j \of \Fam k {\tau_{j,k}}}$, where each $C_j$ expects
a family of arguments at types $\Fam {k \in K_j} {\tau_{j,k}}$, and
a (possibly empty) family of unboxed constructors
$\Fam l {C^\unboxed_l \of \tau_l}$ each expecting a single argument of type
$\tau_l$.

A type $\tau$ is either a type variable $\alpha$, an instance
$\tyconstr t {\Fam i {\tau_i}}$ of a datatype (the $\Fam i {\tau_i}$
instantiate the datatype parameters $\Fam i {\alpha_i}$), or an
instance $\tyconstr {\prim t} {\Fam i {\tau_i}}$ of some primitive
type constructor $\prim t$, such as integers, floats, functions, tuples,
strings, arrays, custom values, etc.

Closed types (without type variables) in this grammar are inhabited by
values $v$ defined by the following grammar, with a simple typing judgment
$v : \tau$ expressing that $v$ has type $\tau$.
\begin{mathpar}
\begin{array}{rcl}
\mathsf{Values} \ni v
& \bnfeq
& C ~ \Fam {k \in K} {v_k} \bnfor C^\unboxed ~ v \bnfor \prim{v}
\end{array}

\infer
{\typedecl
  {\tyconstr t {\Fam i {\alpha_i}}} {\Fam j {C_j \of \Fam k {\tau_{j,k}}}} \dots
 \\\\
 \Fam k {v_k : \subsFam {\tau_{j,k}} i {\alpha_i \leftarrow \tau'_i}}
}
{C_j ~ {\Fam k {v_k}} : \tyconstr t {\Fam i {\tau'_i}}}
\qquad
\infer
{\typedecl
  {\tyconstr t {\Fam i {\alpha_i}}} \dots {\Fam l {C^\unboxed_l \of \tau_l}}
 \\\\
 v : \subsFam {\tau_l} i {\alpha_i \leftarrow \tau'_i}
}
{C^\unboxed_l ~ v : \tyconstr t {\Fam i {\tau'_i}}}

\infer
{\text{language-specific rules for primitive values at primitive types}}
{\prim{v} : \tyconstr {\prim{t}} {\Fam i{\tau_i}}}
\end{mathpar}

\subsection{Low-Level Representation of Values}

Unboxed constructors intrinsically depend on a notion of low-level data
representation.

We assume given a set $\Data$ of low-level
representations, and a function
\begin{mathline}
  \repr : \mathsf{Value} \times \mathsf{Type} \to \Data
\end{mathline}
that determines the data representation of each value.

We further assume two properties of the $\repr$ function:
\begin{enumerate}
\item Injectivity: if $v_1, v_2 : \tau$ have no unboxed constructors
  and $v_1 \neq v_2$, then $\repr(v_1, \tau) \neq \repr(v_2, \tau)$.
\item Unboxing: for any $C^\unboxed \of \tau'$ in $\tau$ and
  $v : \tau'$, we have
  \begin{mathline}
    \repr(C^\unboxed~v, \tau) = \repr(v, \tau')
  \end{mathline}
\end{enumerate}

Our injectivity assumption merely states that our representation was
correct before the introduction of constructor unboxing. Our static
analysis rejects some unboxed constructor definitions to extend this
property to well-typed values with unboxed constructors.

\subsubsection{For OCaml} In the specific case of OCaml,
writing $\machZ$ for the set of machine integers, we claim that the
representation of values in the reference OCaml implementation can be
modeled as:
\begin{mathpar}
\begin{array}{rcl}
\Data_{\text{OCaml}} \ni w & \bnfeq &
  \Imm~(n \in \machZ)
  \bnfor
  \Block ~ {(t \in \machZ)} ~ (a_0, \dots, a_{n-1}) \\
\mathsf{BlockArgs} \ni a & \bnfeq & (w \in \Data) \bnfor (n \in \machZ) \\
\end{array}
\end{mathpar}

As we mentioned in Section~\ref{subsec:ocaml-value-representation},
the low-level representation of OCaml values is either an
\emph{immediate} value, which we approximate as living in $\machZ$, or
a \emph{block} starting with a header word containing a \emph{tag} in
$\machZ$ followed by several words of block arguments
$a_0 \dots a_{n-1}$. Block arguments can be either valid OCaml values
themselves or arbitrary machine words. Note that this representation
loses some information, for example: immediate values live in
a smaller space with one less bit available, and the tag $t$ of
a block determines whether its arguments must be valid OCaml values
(most tags) or machine words (\code{Custom_tag}, \code{String_tag},
\code{Double_tag}, \code{Double_array_tag}, etc.).

We can now define the $\repr$ function going from well-typed source
values to their low-level representation. In OCaml, constant
constructors (taking no argument) are represented as immediates, while
non-constant constructors are represented as blocks, and the
representation of a constructor in each category depends on its
position (indexed starting at $0$) in the type declaration.
\begin{mathpar}
\begin{array}{rcll}
\repr (C ~ \emptyset, \tau) & \eqdef
  & \Imm~i
  & \text{$C$ is $\tau$'s $i$-th constant constructor}
\\
\repr (C ~ \Fam {k \in K} {v_k}, \tau) & \eqdef
  & \Block~i~\Fam K {\repr(v_k, \tau_k)}
  & \text{$C \of \Fam k {\tau_k}$ is $\tau$'s $i$-th non-constant constructor}
\\
\repr (C^\unboxed ~ v, \tau) & \eqdef & \repr (v, \tau')
  & \text{$C^\unboxed \of \tau'$ is an unboxed constructor of $\tau$} \\
\end{array}
\end{mathpar}

As required, this $\repr$ function is injective on boxed constructors
and erases unboxed constructors.

The representations of some primitive values include:
\begin{mathpar}
\begin{array}{rcl}
\repr (\mathsf{true}, \mathsf{bool}) & \eqdef & \Imm~0 \\
\repr ((v_1, v_2), (\tau_1 \times \tau_2)) & \eqdef & \Block ~ 0 ~ (\repr(v_1, \tau_1), \repr(v_2, \tau_2)) \\
\repr (\text{\code{fun x -> x+y}}, \tau_1 \to \tau_2) & \eqdef & \Block ~ {\mathtt{Closure\_tag}} ~ \dots \\
\repr (\mathtt{3.14}, \mathsf{float}) & \eqdef & \Block ~ \mathtt{Double\_tag} ~ \dots \\
\repr (\mathtt{"Hello"}, \mathsf{string}) & \eqdef & \Block ~ {\mathtt{String\_tag}} ~ \dots \\
\end{array}
\end{mathpar}

\subsection{Heads and Head Shapes}

We assume given a set $\mathsf{Head}$ of value \emph{heads}. The head
$h$ of a value $v$ represents an easily computable abstraction/approximation of the
low-level representation of the value $v$: we assume a function
\begin{mathline}
  \head_{\mathsf{data}} : \Data \to \mathsf{Head}
\end{mathline}
computing the head of a value representation, and define
\begin{mathline}
  \head(v, \tau) \quad\eqdef\quad \head_{\mathsf{data}}(\repr(v, \tau))
\end{mathline}
Note that if two values have different heads, then they are
necessarily distinct.

Our static analysis will run on arbitrary type definitions allowed by
our syntax of type declarations, and reject certain type declarations
that would introduce conflicts, that is, allow distinct values with
the same representation. The cleanest way we found to model this was
to define the head shape of a type expression as a \emph{multiset} of
heads, which may contain duplicate elements.

\begin{notation}[$\MSet(S)$, $M(x)$, $\dbraces \ldots$,
  $\msettoset{M}$, $\max(M_1, M_2)$, $M_1 + M_2$]
  We write $\MSet(S)$ for the set of multisets of elements of $S$,
  $M(x)$ for the number of occurrences of $x$ in the multiset $M$,
  $\dbraces \ldots$ for multiset comprehension, and $\msettoset{M}$
  (in $\Set(S)$) for the set of elements of a multiset $M$
  (in $\MSet(S)$).

  We use two standard union-like operations on multisets: the
  maximum and the sum, defined by:
  \begin{mathline}
    \max(M_1, M_2)(x) \eqdef \max(M_1(x), M_2(x))

    (M_1 + M_2)(x) \eqdef M_1(x) + M_2(x)
  \end{mathline}
\end{notation}

We define the \emph{head shape}
$\headshape_{\mathsf{ClosedTypes}}(\tau)$ of a closed type expression
$\tau$ as the multiset of heads of values of this type. We extend this
notation to constructor declarations instantiated at a closed return type,
which we call ``type components'' $\epsilon$ as they come up in type
declarations.
\begin{mathpar}
  \begin{array}{rcl}
    \mathsf{TyComps} \ni \epsilon & \bnfeq & \tau
                                             \bnfor C \of \Fam i {\tau_i} : \tau
                                             \bnfor C^\unboxed \of \tau' : \tau
    \\
  \end{array}

  \begin{array}{lcl}
    \headshape_{\mathsf{ClosedTypes}} & : & \mathsf{ClosedTypes} \to \MSet(\mathsf{Heads}) \\
    \headshape_{\mathsf{ClosedTypes}}(\tau) & \eqdef & \dbraces{\head(v, \tau) \mid v : \tau} \\
    \\
    \headshape_{\mathsf{ClosedTyComps}}(C \of {\Fam i {\tau_i}} : \tau)
      & \eqdef
      & \dbraces{\head(C ~ \Fam i {v_i}, \tau) \mid \Fam i {v_i : \tau_i}} \\
    \headshape_{\mathsf{ClosedTyComps}}(C^\unboxed \of \tau' : \tau)
      & \eqdef
      & \headshape_\mathsf{ClosedTypes}(\tau')
  \end{array}
\end{mathpar}

Finally, we can extend the notion of head shapes to \emph{open types}
(or type components) containing type variables $\alpha$, by taking the
union of all their closed instances. We write $\tau \uparrow \tau'$ if
$\tau'$ is a closed type or type component) that instantiates the free
type variables of $\tau$.
\begin{mathpar}\label{def:headhsape-open-types}
  \infer
  {\epsilon' ~ \text{closed} \\\\ \epsilon' = \subsFam \epsilon i {\alpha_i \leftarrow \tau_i}}
  {\epsilon \uparrow \epsilon'}

  \begin{array}{rcl}
    \headshape_{\mathsf{TyComps}} & :
      & \mathsf{TyComps} \to \MSet(\mathsf{Heads}) \\
    \headshape_{\mathsf{TyComps}}(\epsilon) & \eqdef
      & \max \{ \headshape_{\mathsf{ClosedTyComps}}(\epsilon') \mid \epsilon \uparrow \epsilon' \} \\
  \end{array}
\end{mathpar}

Note that we take the maximum of all closed instances, not their
sum. In particular, if all the closed instances have head shapes that
are sets (they do not contain any duplicates), then the headshape of
the open type is itself a set. In particular, in absence of unboxed
constructors, $\headshape(\alpha)$ is typically equal to the set
$\mathsf{Heads}$ of all heads seen as a multiset (assuming that each
shape is in the image of at least one well-typed value). If we had
used a sum in our definition above, then $\headshape(\alpha)$ would
have duplicates as soon as two distinct types have heads in
common.

\subsubsection{For OCaml}
\label{subsubsec:for-ocaml}

In the specific case of OCaml, we define the head of an immediate $n \in \machZ$ as just the
pair $(\Imm, n)$, and the head of a block of tag $t \in \machZ$ as the
pair $(\Block, t)$.
\begin{mathpar}
\mathsf{Heads}_{\OCaml} \eqdef \{\Imm, \Block\} \times \machZ

\begin{array}{rcl}
\head_{\mathsf{data},\OCaml}(\Imm~n) & \eqdef & (\Imm, n) \\
\head_{\mathsf{data},\OCaml}(\Block~t~(a_0, \dots, a_{n-1})) & \eqdef & (\Block, t) \\
\end{array}
\end{mathpar}

Note that, while our previous choice of $\Data$ and $\repr$ functions
for OCaml model an existing representation, and are not visible to the
users -- even in presence of unboxed constructors -- the definition of
heads and the function $\head : \mathsf{Value} \to \mathsf{Heads}$ are
new design choices that have the user-visible impact of accepting or
rejecting certain unboxed constructors in datatype declarations.

We could use a finer-grained notion of head (for example we could
include the arity $n$ of a block in its head), which allows to
distinguish more types and thus accept more unboxed type definitions.

Conversely, a coarser-grained notion of head would be more portable to
other implementations. For example, an implementation that would
represent constructors by user-visible name rather than position could
not use our notion of head as is. We discuss this portability question
in Section~\ref{subsec:portability}.

Another reason to choose a coarser-grained notion of head is to have
a simpler model to explain to users, at the cost of rejecting more
declarations; for example, one could restrict unboxing to immediate
types by using a pessimistic $\top$ shape for all types containing
blocks.

Finally, our implementation defines a concrete syntax of head shapes
that denote \emph{sets} of heads and is easy to use in
computations. Elements of this \emph{head shape syntax} are pairs of
approximations, one for immediates and one for blocks. Approximations
are defined as either a finite set of machine words (including in
particular the empty set $\emptyset$) or the wildcard shape $\top$
representing all heads.
\begin{mathline}
\begin{array}{rcl}
\mathsf{HeadShapeStx} \ni H & \eqdef & \mathsf{ImmShape} \times \mathsf{BlockShape} \\
\mathsf{ImmShapes}, \mathsf{BlockShapes} & \eqdef & \{ \top \} \cup \FinSet(\machZ) \\
\end{array}
\end{mathline}
\begin{mathline}
\begin{array}{rcl}
\sem \_ & : & \mathsf{HeadShapeStx} \to \Set(\mathsf{Heads}) \\
\sem {(I, B)} & \eqdef & \{ \Imm ~ i \mid i \in \sem I \} \\
              &        & \uplus\; \{ \Block~t \mid t \in \sem B \} \\
\end{array}
~
\begin{array}{rcl}
\sem \_ & : & \mathsf{ImmShapes} \cup \mathsf{BlockShapes} \to \Set(\machZ) \\
\sem \top & \eqdef & \machZ \\
\sem {S} & \eqdef & S \hfill (S \in \FinSet(\machZ)) \\
\end{array}
\end{mathline}

For example, our shape syntax for OCaml integers
is $(\top, \emptyset)$, our shape syntax for OCaml booleans is
$(\{0,1\}, \emptyset)$, our shape syntax for lists or options (one constant
constructor and one non-constant constructor) is
$(\{0\}, \{0\})$, and our shape syntax for custom blocks is
$(\emptyset, \{\mathtt{Custom\_tag}\})$.

This syntax is a correct abstraction of multisets that happen to be
mere sets. It does not let us express multisets with conflicts. We can
directly implement the non-disjoint union $H_1 \cup H_2$ of two
syntactic shapes, and also implement the disjoint union
$H_1 \uplus H_2$ as a partial operation that returns a syntactic shape
if $H_1, H_2$ are disjoint, and is undefined otherwise -- if the
resulting multiset has duplicates, and cannot be represented as
a syntactic shape.
\begin{mathline}
\begin{array}{rcl}
(I_1, B_1) \cup (I_2, B_2) & \eqdef & (I_1 \cup I_2, B_1 \cup B_2) \\
\top \cup x\ ,\ x \cup \top & \eqdef & \top \\
\end{array}

\begin{array}{rcl}
(I_1, B_1) \uplus (I_2, B_2) & \eqdef & (I_1 \uplus I_2, B_1 \uplus B_2) \\
\end{array}
\end{mathline}

In the general case, the definition $\headshape(\tau)$ of head shapes
for open types may be difficult to compute, as it contains
a quantification over all closed extensions of $\tau$. The OCaml value
representation is very regular, which makes it easy to compute shapes
of type variables, constructor declarations and primitive types. We
can represent them directly in our head shape syntax:
\begin{mathpar}
\begin{array}{rcll}
\headshape_\OCaml(\alpha) & \eqdef & {(\top, \top)} & \\
\headshape_\OCaml(C \of \emptyset) & \eqdef & {(\{i\}, \emptyset)}
  & \text{\small $C$ is the $i$-th constant constructor at its type} \\
\headshape_\OCaml(C \of \Fam k {\tau_k}) & \eqdef & {(\emptyset, \{i\})}
  & I \neq \emptyset,\ \text{\small $C$ is the $i$-th non-constant constructor at its type} \\
\headshape_\OCaml(\tyconstr {\prim t} {\Fam i {\tau_i}}) & \eqdef & 
  & \text{\small the immediates and tags of primitive type constructor $\prim t$} \\
\end{array}
\end{mathpar}

This definition of $\headshape_\OCaml$ on base types and boxed
constructor definitions agrees with the generic definition
$\headshape$, in the sense that
$\headshape_{\mathsf{TyComps}}(\epsilon) = \sem{\headshape_\OCaml(\epsilon)}$
for the OCaml value representation and our choice of heads.

\subsection{Sum Normal Form}
\label{subsec:sum-normal-form}

To compute the head shape of a type expression $\tau$, we must unfold
type definitions and traverse unboxed constructors. This
transformation is of independent interest, we formalize it in this
section.

We define a grammar of \emph{sum normal forms} $S$ that capture the
result of this unfolding process, and a (partial) normalization
judgment $\jnorm \tau S$ that computes the head normal form of
a type.
\begin{mathpar}
  \begin{array}{lcl}
    S & \bnfeq & \emptyset \mid \eta \mid S + S \\
    \eta & \bnfeq & \alpha
                    \mid C \of {\Fam i {\tau_i}}
                    \mid \tyconstr {\prim{t}} {\Fam i {\tau_i}} \\
  \end{array}

  \infer[var]
  { }
  {\jnorm \alpha \alpha}

  \infer[prim]
  { }
  {\jnorm
    {\tyconstr {\prim{t}} {\Fam i {\tau_i}}}
    {\tyconstr {\prim{t}} {\Fam i {\tau_i}}}}

  \infer[constr]
  {\typedecl
    {\tyconstr t {\Fam i {\alpha_i}}}
    {\Fam j {C_j \of {\Fam k {\tau_{j,k}}}}}
    {\Fam l {C^\unboxed_l \of \tau_l}}
   \\
   \Fam l {\jnorm {\subsFam {\tau_l} i {\alpha_i \leftarrow \tau_i}} S_l}
  }
  {\jnorm
    {\tyconstr {t} {\Fam i {\tau_i}}}
    {\sum_j {C_j \of {\Fam k {\subsFam {\tau_{j, k}} i {\alpha_i \leftarrow \tau_i}}}}
      + \sum_l S_l}}
\end{mathpar}

A sum normal form is a multiset of components $\eta$ written as
a formal sum, that are either a boxed constructor, a type variable or
a primitive type constructor. The normalization judgment unfolds
datatype declarations, sums boxed constructors and the normal form of
the unboxed arguments.

In presence of mutually-recursive definitions, some type expressions
may ``loop'' forever: they don't have a sum normal form. Consider for
example:
\begin{mathline}
  \typedecl {\mathsf{loop}} {} {\mathsf{Int}^\unboxed \of \prim{int} \mid \mathsf{Loop}^\unboxed \of \mathsf{loop}}
\end{mathline}
Fortunately, the problem of whether a given type expression $\tau$ has
a sum normal form is in fact decidable. We discuss our decision
procedure in Section~\ref{sec:halting-problem}.

\subsection{Rejecting Conflicts}

Finally, given a type declaration $d$, our static analysis computes
a head shape $\headshape(S)$ of its sum normal form $S$ by summing the
head shape of each component of the sum. Our static analysis accepts
the definition if and only if $\headshape(S)$ does not contain any
duplicates.
\begin{mathpar}
  \infer
  {\jnorm {\tyconstr t {\Fam i {\alpha_i}}} S}
  {\headshape_{\mathsf{decl}}(\typedecl {\tyconstr t {\Fam i {\alpha_i}}} \dots {}) \eqdef \headshape_{\mathsf{snf}}(S)}
\end{mathpar}
\begin{mathpar}
  \begin{array}{lcl}
    \headshape_{\mathsf{snf}}(\emptyset) & \eqdef & \emptyset \\
    \headshape_{\mathsf{snf}}(\eta) & \eqdef & \headshape_{\mathsf{TyComps}}(\eta) \\
    \headshape_{\mathsf{snf}}(S_1 + S_2) & \eqdef & \headshape_{\mathsf{snf}}(S_1) + \headshape_{\mathsf{snf}}(S_2) \\
  \end{array}
\end{mathpar}

The result of this analysis can be easily computed, in the case of
OCaml, by using our head shape syntax: the head shape of $S_1 + S_2$
is conflict-free if
$\headshape_\OCaml(S_1) \uplus \headshape_\OCaml(S_2)$ is defined and
has a conflict otherwise.

\subsection{Pattern-Matching Compilation}

When checking a type declaration with unboxed constructors, we record
for each unboxed constructor $C^\unboxed \of \tau$ the head shape of
its type parameter $\tau$. When compiling pattern-matching clauses
using an unboxed constructor in a pattern, say $C^\unboxed p$, we
implement matching on $C^\unboxed \of \tau$ as a condition that the
head of the scrutinee must belong to the shape of $\tau$. (The details
of how to do this depends of course on the pattern-matching
compilation algorithm of the language.)

We know that this approach is always sound, thanks to the property
that none of the other scrutinees (starting, at the source level, with
a different constructor) may have a head belonging to the head shape
of $\tau$. Note that this property, enforced by our static analysis,
is in fact slightly stronger than the absence of conflicts: not only
must the representation of inhabitants of $\tau$ be distinct from all
the other possible scrutinees, they should furthermore have distinct heads.

The runtime cost of checking the head depends on the language and the
notion of head chosen. For our choice of heads for OCaml, it is
exactly as costly as checking the head constructor of a value, so this
does not add overhead on pattern-matching. A finer-grained notion of
head that would inspect the value ``in depth'' could add a higher
cost -- to balance against the space savings of accepting more
unboxing.

\section{Scaling to a Full Language}
\label{sec:scaling}

In this section, we describe less formally all the ``other issues''
that we had to consider to scale this proposed feature to a full
programming language, namely OCaml.

\subsection{Handling All the Tricky Cases}
\label{subsec:tricky-cases}

We did not encounter any conceptual issue when scaling this approach
to all the primitive types supported by the OCaml runtime. (This is
not too surprising given that our heads are closely modeled on the
existing runtime data representation.) In the interest of
demonstrating the difference between the simple situation of datatypes
with a simple representation and everything else, let us give here an
\emph{exhaustive} list of all the tricky cases.

\begin{enumerate}
\item \code{Double_array_tag} is used to represent nominal records
  whose fields are all \code{float}, and also values of type
  \code{float array}.

  To determine the shape of a record type, we must call the same logic
  that the type-checker uses to the decide the ``unboxed float record''
  criterion and use $(\emptyset, \{\mathtt{Double\_array\_tag}\})$
  instead of $(\emptyset, \{0\})$ in that case.

  For arrays, we use the head shape
  $(\emptyset, \{0, \mathtt{Double\_array\_tag}\})$ in all cases. Note
  that OCaml supports a configuration option to disable the unboxing
  of float array (supporting this representation adds some dynamic
  checks on array operations), but we decided to use the pessimistic
  shape with both tag values independently of the configuration value,
  to avoid having the compiler statically reject some programs only in
  some specific configuration.

\item Values of type \code{ty Lazy.t} have an optimized representation
  where they may be represented by a lazy thunk of tag
  \code{Lazy_tag}, a computed value of tag \code{Forward_tag},
  \emph{or} directly a value of type \code{ty}, under some conditions
  on \code{ty}. The corresponding shape is the maximum of
  $(\emptyset, \{\mathtt{Lazy\_tag}, \mathtt{Forcing\_tag}, \mathtt{Forward\_tag}\})$
  and of the shape of \code{ty}. (OCaml 5 added a third tag
  \code{Forcing_tag} to detect concurrent forcing; it was trivial to
  adapt our analysis.)

  Note that Section~\ref{subsubsec:for-ocaml} defined
  $\headshape_\OCaml{\tyconstr {\prim t} {\Fam i {\tau_i}}}$ as
  depending only on $\prim t$, not the $\tau_i$; here we are handling
  lazy values in a more precise way due to their non-uniform
  representation. (We could also approximate them to the uniform shape
  $\top$.)

\item Function closures may have either tags \code{Closure_tag} or
  \code{Infix_tag} (used for some mutually-recursive functions). The
  pattern-matching code that we generate on sum types with an unboxed
  function type (which is useful for Coq \code{native_compute}) is
  slightly less good than it could be because these two tags are not
  consecutive (247 and 249), so we are slightly tempted to renumber
  the tags in the future.
\item Exceptions, and in general inhabitants of extensible sum types,
  use tag \code{Object_tag}. Object types themselves have an obscure
  and complex data representation due to various optimizations, and we
  just assigned them the top shape $\top$.
\end{enumerate}

\subsection{Portability of our Heads}
\label{subsec:portability}

The language of head shapes makes some aspects of the low-level
representation of values visible to users of the surface
language. This comes at the risk of complexity, but also at the risk
of reducing the portability of the language by setting in stone
certain representation choices, that would rule out other
implementations.

One could think of making constructor-unboxing a ``best effort''
feature to avoid this downside, by simply emitting a warning in the
case where an unboxing annotation would introduce a conflict under the
current implementation, and keeping the constructor boxed. We decided
against this for now, because we believe that advanced performance
features such as constructor unboxing are used when users reason about
the performance of their application, that is, when data
representation is part of their specification for the code they are
writing. In this context, silently ignoring representation requests is
arguably a bug: it breaks the specification the user has in mind.

Instead we are trying to discuss with other implementors of OCaml to
find whether we should make our heads more coarse-grained in some
places, to increase portability without breaking relevant examples of
interest. In particular, the alternative backend \code{js_of_ocaml}
compiles OCaml to JavaScript, and uses native JavaScript numbers for
most OCaml numeric types (\code{int}, \code{float}, \code{nativeint},
\code{int32}). We are planning to quotient the difference between
those types in our language of shape, to improve portability.

\subsection{Abstract Types with Shapes}
\label{subsec:abstract-with-shapes}

Abstract types can readily be given the top shape $\top$. We also
support annotating an abstract type with a shape restriction
\code{[@shape ..]}, which gives a head shape for this type.

For abstract types coming from the interface of modules or functor
parameters, those shape annotations are checked when checking that the
interface conforms to the implementation.

For abstract types used to represent values only populated by the FFI,
these shape annotations have to be trusted, in the same way that the
OCaml FFI trusts foreign functions to respect their type provided on
the OCaml side.

We used this feature in our Zarith example in
Section~\ref{subsec:unboxed-zarith} to allow unboxing the \code{Big}
constructor, whose argument is an abstract type \code{custom_gmp_t} of
GMP numbers implemented through the FFI.

\subsection{Are We Really First-Order?}

Parametrized type definitions fall in the first-order fragment because
OCaml does not support higher-kinded types.

Note that some designs for higher-kinded types in related languages
are restricted to higher-order ``type constructors'' that do not
create $\beta$-redexes, so they do not necessarily have the
expressiveness of the full higher-order $\lambda$-calculus.

On the other hand, the OCaml module system does provide higher-order
abstractions through functors: a type in a functor may depend on
a parametrized type in the functor argument. However, unfolding of
type definitions remain first-order in nature:
\begin{itemize}
\item When we are checking the body of a functor and encounter a type
  that belongs to a module parameter, it is handled as any other type
  declaration.
\item When we encounter a type expression containing a functor application,
  e.g.~\code{Set.Make(Int).t}, the type-checker has access to the signature
  of the functor application \code{Set.Make(Int)} and we check its type \code{t}.
\end{itemize}

Another way to think of the treatment of functor application is that
the OCaml type-checker performed $\beta$-reduction of functor
applications before we compute shapes. In other words, in this work,
we consider the module language as a strongly-normalizing higher-order
subset whose normal forms are first-order.

\subsection{OCaml Features Subsumed by Head Shapes}

The OCaml type-checker currently contains three subsystems that we
believe would be subsumed by our head shape analysis:
\begin{enumerate}
\item It contains an analysis of the ``unboxed form'' of a type
  (due to the presence of unboxing for single-constructor variants and
  single-field records) that corresponds to our notion of sum normal
  form of a type, and would benefit from our normalizing algorithm to
  compute those in presence of recursion.
\item It defines a property called \code{[@@immediate]} for abstract
  types, which claims that the inhabitant of the type are all
  immediate values. (This is used by the runtime to specialize ad-hoc
  polymorphic functions such as comparison and serialization.)
  Computing the head shape subsumes immediateness-checking.
\item To implement unboxing of single-variant GADTs, it must perform
  an intricate static analysis to reject attempts to unbox
  existentials, which would make the \code{float array} optimization
  unsound. (Don't ask.) We believe that this static analysis, detailed
  in~\citet*{mutual-unboxing}, could be subsumed by our head shape
  computation.
\end{enumerate}

We are planning to simplify the compiler implementation by removing
all the existing logic to implement these separate aspects, and
replace them by a shape computation.

\section{Potential Extensions}

We have considered the following aspects, but have not implemented
them. They are not necessary to consider upstreaming a first useful
version of constructor unboxing.

\subsection{Shape Constraints on Type Variables}
\label{subsec:shape-constraints-on-type-variables}

Section~\ref{subsec:abstract-with-shapes} shows how abstract types
can now be annotated with constraints on the head shapes of their
inhabitants. A related feature would be to constrain the type
variables of parametrized types:

\begin{lstlisting}
type ('a [@shape any_block]) block_option = None | Some of 'a [@unboxed]
type ('a [@shape immediate]) imm_option = None of unit | Some of 'a [@unboxed]
\end{lstlisting}

The type \code{'a block_option} is similar to the standard
\code{'a option} datatype, but its type parameter \code{'a} may only
be instantiated with type expressions whose shape is included in the
shape \code{any_block} -- any block tag, but no immediate value. This
restriction allows unboxing the \code{Some} constructor without
conflicting with the \code{None} immediate value.
Conversely, \code{imm_option} may only be instantiated with immediate
types, and its \code{Some} constructor can also be unboxed as we made
\code{None of unit} a block constructor.

Similarly to abstract types with shapes
(Section~\ref{subsec:abstract-with-shapes}), this feature provides
modularity. We can construct large types with specialized
representations by composing together smaller parameterized types
(or functors), with shape assumptions on the boundaries between the
various definitions.

This change is easy conceptually, but requires non-trivial
changes to the OCaml compiler where type variables do
not carry kind information. It goes in the same direction as other
``layout'' changes experimented with by Jane Street, so some of the
implementation work can be shared.

\subsection{Harmless Cycles}

Our algorithm to compute shapes unfolds potentially-recursive type
definitions and monitors termination: it stops when encountering
a cycle in the definition. Currently our prototype rejects all
definition that contain such cycles. But the cycles fall in two
categories: most cycles are ``harmful cycles'' that must be rejected,
but there are ``harmless cycles'' that could be accepted.


\begin{lstlisting}
type harmful = A | Loop of harmful [@unboxed]
type harmless = Loop of harmless [@unboxed]
\end{lstlisting}

Both those examples are rejected by our prototype as their shape
computation detects a cycle. \code{harmful} cannot soundly be
accepted, as there would be a confusion between the values \code{A},
\code{Loop(A)}, \code{Loop(Loop(A))}, etc. On the other hand, allowing
the unboxing of \code{harmless} would not in fact introduce any
confusion as the type would be empty -- without any inhabitant. Said
otherwise, cycles in shape computations can be interpreted as smallest
fixpoints; most of those fixpoints contain conflicts but a few are the
empty set of value.

Accepting harmless cycles should be of medium difficulty. From an
implementation perspective it is not easy to distinguish harmless
cycles and accept them, it is substantially more work than rejecting
all cycles. Besides, there is no point in writing types such as
\code{harmless} in practice -- just write an empty type directly. So
this is naturally left as future work.

There is however one good reason to do more work there, which is
related to data abstraction. Consider the following example:
\begin{lstlisting}
type 'a foo                     type weird = Loop of weird foo [@unboxed]
\end{lstlisting}

This \code{weird} definition is accepted by our shape analysis. For
the abstract type \code{'a foo} we assume the shape $\top$ of any
possible value -- this does not depend on the parameter
\code{'a}. Then \code{weird foo} has the same shape $\top$ and the
definition of \code{weird} is accepted. (Adding any other constructor
to \code{weird} would make it rejected.)

However, we could later learn that the type \code{'a foo} is in fact
defined as \code{type 'a foo = 'a}. If we perform the substitution, we
get a harmless cycle. In other words, rejecting harmless cycles breaks
the substituability property for abstract types. This is a nice
meta-theoretical property, and breaking it may result in surprising
software engineering situations that are problematic in practice.

\subsection{Unboxing by Transformation}
\label{subsec:unboxing-by-transformation}

In our work, unboxed constructors act as the identity on the
representation of their arguments. One could generalize this by
allowing constructor unboxing to be realized by a non-identity
transformation on its arguments -- chosen to be more efficient than
the default constructor representation. Applying non-identity
transformations could avoid conflicts in value representations,
allowing more unboxing requests. Consider for example:
\begin{lstlisting}
type 'a t = A of bool | B of 'a option
\end{lstlisting}

Our work only supports unboxing the constructor \code{A} in this
example. Unboxing \code{B} is not supported:
\begin{itemize}
\item if \code{A} is not unboxed, then we would have a confusion
  between blocks constructed by \code{A} and those coming from the
  \code{Some} constructor of the option.
\item if \code{A} is unboxed, then we would have a confusion between
  immediates corresponding to the \code{false} value (in \code{A}) and the
  \code{None} value (in \code{B}).
\end{itemize}

It would however be possible to unbox the constructor \code{B} if we
accepted to change the representation of the constructor \code{A of
  bool [@unboxed]}. Instead of storing a \code{bool} value directly
(an immediate in $\{0, 1\}$), we could transform the \code{bool} value
to store it as an immediate in $\{1, 2\}$ for example, avoiding
a conflict with the \code{None} value (immediate $0$) unboxed from
\code{B}. Examples of representative approaches follow.

\begin{lstlisting}
(* no shifting, but a different boolean type *)
\end{lstlisting}\vspace{-.5em}
\begin{minipage}{0.5\linewidth}
\begin{lstlisting}
type 'a t1 =
  | A of fake_bool [@unboxed]
  | B of 'a option [@unboxed]
\end{lstlisting}
\end{minipage}
\begin{minipage}{0.5\linewidth}
\begin{lstlisting}
and fake_bool =
  | Fake_false [@tag 1]
  | Fake_true [@tag 2]
\end{lstlisting}
\end{minipage}

The type \code{t1} is in fact not an example of a transformation
associated with an unboxed constructor: instead we assume that it is
possible to specify a non-standard choice of tag at declaration
time -- the imaginary \code{[@tag 2]} attribute. Supporting this would
be an easy change, but it requires using a non-standard boolean type
and thus requires code changes for the user, making it cumbersome or
impractical in many situations.

\begin{minipage}{0.5\linewidth}
\begin{lstlisting}
(* explicit shifting *)
type 'a t2 =
  | A of bool [@unboxed by (add 1)]
  | B of 'a option [@unboxed]
\end{lstlisting}
\end{minipage}
\begin{minipage}{0.5\linewidth}
\begin{lstlisting}
(* fully inferred transformation *)
type 'a t3 =
  | A of bool [@unboxed]
  | B of 'a option [@unboxed]
\end{lstlisting}
\end{minipage}

The type \code{t2} performs a transformation at unboxing time
(adding \code{2} to the value) that is specified by the
user. Pattern-matching code would then have to be careful to undo this
transformation on the fly (by subtracting \code{1}). Note that
\code{(add 1)} is not an arbitrary OCaml term here, it must be part of
a dedicated transformation DSL that we know to invert efficiently.

Another valid choice would be to have the constructor \code{B}
transform its argument in a way that leaves its block values unchanged
but shifts its \code{None} value from the immediate $0$ to an
immediate outside $\{0,1\}$ or to a block of non-zero tag
(with no argument). Note that unboxing \code{B} would come at a higher
runtime cost as the transformation (to apply at construction time and
unapply at matching time) is more complex.

Finally the type \code{t3} assumes a version of constructor unboxing
that implicitly infers such transformations to satisfy the user's
unboxing request. This is not the design approach that we have used in
our OCaml work, but it corresponds to unboxing strategies in some
other programming languages -- see our discussion of Rust's
niche-filling optimizations in our Related Work
Section~\ref{subsec:rust-niche-filling}.

Arbitrary transformations can be supported, as long as we can express
the corresponding abstract transformation on shapes. For OCaml and the
notion of heads that we proposed, a natural space of transformations
are those that would change the head of a value, and leave the rest
unchanged, by:
\begin{enumerate}
\item applying a mapping to its immediate values, and/or turning some
  of them into blocks (constant blocks with fixed tags, or
  non-constant blocks with the transformed immediate as argument)

\item applying a mapping to the tag of its blocks, leaving its arity and
  arguments unchanged
\end{enumerate}

The set of transformations of interest is also constrained by
performance considerations. In particular, turning an immediate into
a non-constant block requires an allocation and memory indirection
(constant blocks can be preallocated), which is precisely what we
wanted to avoid by unboxing the constructor. It may still be
beneficial if this transformation occurs only for some inputs that are
rare in practice, with all other cases unboxed.

Supporting tag choice requests as in \code{t1} should be easy;
user-specified transformations as in \code{t2} would be of medium
difficulty, depending on the expressiveness of the transformations. We
are not planning to work on full transformation inference in the
context of OCaml.

\subsection{Using Unboxing to Describe Existing Representation Tricks}
\label{subsec:unboxing-existing-tricks}

Some subtle data-representation choices of the OCaml compiler and
runtime, mentioned in Section~\ref{subsec:tricky-cases}, could in fact
be presented as unboxing, possibly with further extensions.

\paragraph{Flat float arrays} OCaml arrays use a uniform
representation (using tag \code{0}) except for arrays of elements
represented as floats, which have the tag \code{Double_array_tag} and
a custom representation.  The OCaml runtime (written in C) checks the
array tag on each low-level operation to determine how to access the
array.
OCaml cannot currently express the type of custom arrays of float-represented
values, but our proposed shape annotations
(Section~\ref{subsec:shape-constraints-on-type-variables}) would make
it possible to do so:

\begin{lstlisting}
type ('a [@shape double]) double_array [@@shape double_array]
\end{lstlisting}

\noindent
With constructor unboxing we can then express the array
representation trick in safe OCaml code:

\begin{lstlisting}
type 'a array =
 | Any : 'a generic_array -> 'a array
 | Double : ('d [@shape double]). 'd double_array -> 'd array
\end{lstlisting}

The type \code{'a generic_array} is a type that does not exist in
OCaml today, of uniform arrays with
tag \code{0}.  The constraint \code{'d [@shape double]} indicates that
(in our proposed OCaml extension) matching a value of type \code{'d
array} with the \code{Double} constructor reveals that \code{'d} is
represented as \code{double}.
This definition would suffice for defining array-accessing
functions in pure OCaml, but for array creation the runtime checks
dynamically if its argument is a \code{float}. Implementing
the check in OCaml would require exposing an extra (inelegant and
non-parametric, but safe) primitive:

\begin{lstlisting}
type 'a double_check =
 | IsAny : 'a double_check
 | IsDouble : ('d [@shape double]). 'd double_check
val check_if_double : 'a -> 'a double_check
\end{lstlisting}

\paragraph{Lazy forwarding} The representation of a lazy value may be
a block of tag \code{Lazy_tag}, for a thunk that has not yet been
evaluated to a result, or a block of tag \code{Forward_tag} storing
a result, or sometimes this resulting value directly. The OCaml
runtime sometimes ``shortcuts'' forward blocks when they are moved
around, replacing them by their value directly, with a dynamic
check that this does not introduce an ambiguity -- the value should
not itself have tag \code{Lazy_tag} or \code{Forward_tag}.

Let us try to express this shortcutting trick in OCaml rather than in
the runtime code in C. Let us assume an imaginary \code{[@unboxed unsafe]}
attribute that does not perform any static confusion check for this
constructor -- we intentionally do not provide this in our current
design proposal, which focuses on safe uses. We could then write:

\begin{lstlisting}
type 'a lazy_state =                             and 'a lazy = 'a lazy_state ref
| Thunk of (unit -> 'a) [@tag lazy_tag]
| Forward of 'a         [@tag forward_tag]
| Forward_unboxed of 'a [@unboxed unsafe]

let make_forward (v : 'a) : 'a lazy_state =
  if List.mem (Obj.tag (Obj.repr v)) Obj.[lazy_tag; forward_tag; double_tag]
  then Forward v
  else Forward_unboxed v
\end{lstlisting}

One could even think of an imaginary \code{[@unboxed dynamic]} variant where the compiler is in charge of inserting this dynamic check:

\begin{lstlisting}
type 'a lazy_state =
| Thunk of (unit -> 'a) [@tag lazy_tag]
| Forward of 'a         [@tag forward_tag] [@unboxed dynamic]
\end{lstlisting}

\section{Our Halting Problem}
\label{sec:halting-problem}

In Section~\ref{subsec:sum-normal-form}, we mention that computing the
head shape of a type requires unfolding datatype definitions, and that
this unfolding process may not terminate in presence of
mutually-recursive datatype definitions.

In the present section, we discuss this problem in more detail, and
present a novel algorithm to normalize safely in presence of
recursion: it either returns the sum normal form of a type, or reports
(in a finite amount of time) that the definition loops and no sum
normal form exists.

First, we remark that this problem corresponds to the \emph{halting
  problem} for a specific fragment of the pure $\lambda$-calculus
(just function types, no products, booleans, natural numbers etc.),
namely the first-order fragment with arbitrary recursion. Consider the
following example:

\begin{minipage}{0.5\linewidth}
\begin{lstlisting}
type 'a id = 'a
type name = Name of string [@unboxed]
\end{lstlisting}
\end{minipage}
\begin{minipage}{0.5\linewidth}
\begin{lstlisting}
type handle =
  | By_number of int id [@unboxed]
  | By_name of name [@unboxed]
  | Opaque of string
\end{lstlisting}
\end{minipage}

It can be rephrased as a $\lambda$-term with recursive definitions as follows:
\begin{lstlisting}
let rec id(a) = a
and name = string
and handle = sum (id int) (sum name (box string))
\end{lstlisting}
In this translation, we use a free variable \code{sum} as a binary
operator to separate constructor cases, and a free variable \code{box}
over the translation of type expressions appearing under
a constructor. (Other free variables encode primitive types.) The sum
normal form of any type in the definition environment above can be
read back from a normal form of the translated $\lambda$-terms in
presence of the recursive definitions. (More precisely, we only need
a ``weak'' normal form that does not reduce under \code{box}
applications.)

The algorithm that we present in this section decides the halting
problem for the first-order pure $\lambda$-calculus with recursive
definitions, also called ``order-1 recursive program schemes'' in the
literature, with an arbitrary reduction strategy. This is not
unreasonable, given that the halting problem for this fragment is
already known to be decidable, as demonstrated for example
in~\citet*{first-order-halting-problem} in the first-order case and
in~\citet*{higher-order-halting-problem,lmcs:1567} (for example) in
the more general setting of the pure simply-typed $\lambda$-calculus with
arbitrary recursive definitions!

\subsection{On-The-Fly rather than Global Termination Checking}

The normalization arguments in previous work on recursive program
schemes are \emph{global} in nature; they reason by normalizing all
mutually-recursive definitions at
once~\citep*{first-order-halting-problem}, or at least they compute
a termination bound that depends on the size of the whole
mutually-recursive system~\citep*{higher-order-halting-problem}.

In our setting, the set of mutually-recursive definitions potentially
contains large type definitions in scope, that are either explicitly
mutually-recursive, or depend on each other through recursive
modules. Normalizing datatype definitions without unboxed constructors
is immediate as they are their own normal forms, but we also have to
normalize through OCaml type \emph{abbreviations} which are widely
used. (We have not included type abbreviations in our
Section~\ref{sec:head-shapes} as the abbreviation \code{type
  t = $\ \tau$} can be understood in this context as syntactic sugar
for \code{type t = Abbrev$^\unboxed\ $ of $\ \tau$}.) Expanding all
abbreviations is also known to potentially generate very large
structural types for some use-cases, so the OCaml type checker uses
careful memoization to only expand on-demand during type inference.

Another issue with a global termination analysis is that OCaml type
definitions \emph{change often} due to functor
applications.\footnote{Functor in the ML-family sense of a module
  (possibly carrying type components) parameterized by another module.}
Some type definitions in a functor body rely on abstract (or concrete)
types from the functor argument. When the functor is applied to
a module parameter, we get a new instance of those definitions where
previously abstract type constructors from the formal argument are
concrete, which may even introduce new recursive dependencies. Any
global termination analysis done on all type definitions would thus
have to be partially recomputed on functor applications -- and
delimiting the part of the computations to rerun may not be obvious
in presence of recursive dependencies.

Performing a global termination analysis thus runs the risk of large
computational costs in practice, which is all the more frustrating
that we expect unboxed constructors themselves to be rarely used,
being an advanced feature. It should come at no cost when not used and
at little cost when used sparingly.

Instead, we propose an \emph{on-the-fly} termination checking
algorithm. Without any static precomputation on the set of
mutually-recursive definitions (which may be large and/or change often
during type checking), our algorithm takes a term and monitors its
reduction sequence: it maintains some information on the side that is
updated during reduction, and may ``block'' the reduction if it
detects that it is about to loop forever. We must provide the
following guarantees:
\begin{itemize}
\item Correctness: reduction sequences that are never blocked by the
  termination monitor are always finite; they cannot diverge.
\item Completeness: if a reduction sequence is blocked by the termination monitor,
  then it would have diverged in absence of monitoring.
\end{itemize}

The existence of a termination monitor that is sound and complete
implies decidability of the halting problem for the reduction being
considered. We have not found termination results in the existing
literature whose proofs would suggest this termination-monitoring
approach; to the best of our knowledge, this approach is novel for the
pure first-order $\lambda$-calculus with recursive definitions. (But the
literature on term rewriting systems is vast and our knowledge of it
very partial.)

We consider this as a notable contribution of our work whose interest
is independent of constructor unboxing. Already in the OCaml compiler,
there are other parts of the type checker that need to normalize type
definitions (including unboxed single-constructor datatypes), and rely
on the unprincipled approach of passing a fixed amount of ``fuel'' and
failing with an error once it is exhausted. We plan to rewrite these
computations using on-the-fly normalization checking. We hope that
other language implementors could use this approach to work with
recursive type definitions, and there may be other use-cases thanks to
the generality of the language considered.

Remark: our termination-monitoring approach means that we are only
adding some bookkeeping logic to a head-shape computation that we need
to do anyway, that happens rarely (single-constructor unboxing is
a rarely used, opt-in feature), and whose cost is bounded by a very
small constant in practice, the depth of type definitions that need to
be unfolded to compute the head shape (at most 5 for reasonable
OCaml programs). In particular, we know with certainty that head shape
computations will add no noticeable compile-time overhead to the
compilation process of real-world OCaml programs.

\subsection{Intuition}

\paragraph{Attempt 1: detect repetition of whole terms}

A first idea to prevent non-termination is to perform a simple cycle
detection: block the reduction sequence if we encounter a term that
was already part of the reduction sequence. This approach is obviously
complete: we have found a cycle that can diverge, but it is not correct in
presence of ``non-regular recursion'', that can generate infinitely
many distinct terms. Consider for example (in $\lambda$-calculus syntax):
\begin{lstlisting}
let rec loop(a) = loop(list a) in loop(int)
\end{lstlisting}
In this environment \code{loop(int)} reduces to
\code{loop(list int)}, then \code{loop(list (list int))}, etc.,
without ever repeating the same term or even the same subterm in
reducible position.

\paragraph{Attempt 2: detect repetition of head functions}

The second idea is that, if detecting repeating of whole terms is too
coarse-grained, we should instead track repetitions of the
\emph{head function} of the term, in the usual sense of the topmost
function/rule/constructor in reducible position. This approach would
prevent the infinite reduction sequence for \code{loop int} in the
example above, by blocking at the second redex with the same head
\code{loop}. It is easy to show that it is sound for termination,
given that the number of distinct heads is finite. However, this
approach is incomplete, it blocks reduction sequences that would have
normalized, for example:
\begin{lstlisting}
let rec id(a) = a in id (id int)
\end{lstlisting}
This reduction sequence needs to reduce the function \code{id} twice
before reaching a normal form.

\paragraph{Solution: trace head functions for each subexpression}

Our solution is a refinement of detecting repetition of heads. Instead
of tracking the heads that have been expanded in the whole term, we
trace a \emph{different} set of heads for each subterm, corresponding
to the set of function heads whose expansions were necessary to have
the subterm appear in the term.

In the first example above, we start with the term \code{loop int}
where all subterms are annotated with the empty trace
(no expansion happened). We can write this as \code{[]loop []int}: the
subterms \code{loop int} and \code{int} are both in the empty
trace. The first step of the reduction results in the annotated term
\code{[loop]loop ([loop]list []int)}: the subterm \code{[]int} was
unchanged by the expansion, but the surrounding context %
\code{loop (list $\ \hole$)} \emph{appeared} in the reduction of
\code{loop} in the empty trace, so this part of the new term gets
annotated with \code{[loop]}. At this point, our algorithm blocks the
redex \code{loop (list int)} as it is already annotated with the
function head \code{[loop]}.

In the second example, \code{[]id ([]id []int)} reduces to its
argument \code{[]id []int} unchanged -- this argument was already
present in the term before the expansions, it did not \emph{appear}
during the reduction. \code{[]id []int} can in turn reduce to
\code{[]int} which is an (annotated) normal form.

\subsection{Formalizing our Algorithm}
\label{subsec:first-order-calculus}

We start from a grammar for programs $p$ in the first-order
$\lambda$-calculus with recursive definitions, containing in
particular terms $t$, and introduce a distinct category of
\emph{annotated} terms $\bar t$ whose function-call subterms carry an
expansion trace $l$ (a list of function names without duplicates).
\begin{mathline}
  \begin{array}{lcll}
    p & \bnfeq & \letrecin D t & \text{recursive programs} \\
    D & \bnfeq & \emptyset \bnfor D, f {\Fam i {x_i}} = t & \text{function definitions} \\
    f & & & \text{function name} \\
    x & & & \text{first-order variable} \\
    t & \bnfeq & x \bnfor f {\Fam i {t_i}} & \text{term} \\
  \end{array}
\end{mathline}
\begin{mathline}
  \begin{array}{lcll}
    \bar p & \bnfeq & \letrecin D {\bar t}                    & \text{annotated programs} \\
    \bar t & \bnfeq & x \bnfor \loc {f {\Fam i {\bar t_i}}} l & \text{annotated term} \\
    l & \bnfeq & \emptyset \bnfor l, f                        & \text{(we require $f \notin l$)}\\
    C[\hole] & \bnfeq & \hole \bnfor \loc {f {(\Fam i {\bar t_i}, C[\hole], \Fam j {\bar t_j}})} l
      & \text{annotated reduction context} \\
  \end{array}
\end{mathline}

We extend the usual notion of $\beta$-reduction $t \rewto t'$ to
annotated terms. Expanding a function call
$f {\Fam i {\bar t_i}}$ is only possible if its trace does not already
contain $f$ -- otherwise this term is stuck, we call it
a \emph{blocked redex}. The arguments $\Fam i {\bar t_i}$ are
annotated terms, but the body $t'$ of the definition of the function
$f$ is a non-annotated term: the body and its subterms \emph{appear}
in the reduction sequence at this point, and we use an
\emph{annotating} substitution $\locsubs {t'} \sigma l$ to annotate
them, where $\sigma$ is a substitution from variables to annotated
terms and $l$ is the trace to use to annotate new subterms.

\begin{definition}[$\bar p \rewto \bar p'$, $\locsubs t \sigma l$]
\begin{mathpar}
\infer
{(f {\Fam i {x_i}} = t') \in D \\ f \notin l}
{
 \letrecin D C[\loc {f {\Fam i {\bar t_i}}} l]
  \rewto
  \letrecin D C[\locsubs {t'} {\Fam i {x_i \leftarrow \bar t_i}} {l,f}]
}

\begin{array}{rcl}
\locsubs x \sigma l
& =
& \sigma(x)
\\
\locsubs {f {\Fam i {t_i}}} \sigma l
& =
& \loc {f {\Fam i {\locsubs {t_i} \sigma l}}} l
\end{array}
\end{mathpar}
\end{definition}

Note that the annotating substitution $\locsubs t \emptyset l$
annotates each function call of an unannotated term $t$ with the trace
$l$.

\begin{notation}[$\floor {\bar t}$]
  Let us write $\floor {\bar t}$ (or $\floor {\bar p}$) for the
  unannotated term (or program) obtained from erasing all annotations
  from $\bar t$ (or $\bar p$).
\end{notation}

We are only interested in annotated terms that were obtained starting
from an initial annotated term with all traces empty. Outside this
subset of annotated terms there are terms with weird/impossible
annotations (for example annotation of functions that do not exist in
the recursive environments) that we sometimes want to rule out from
our statements.

\begin{definition}[Reachable annotated term]\label{def:reachable}
  An annotated term $\bar t$ is \emph{reachable} if it occurs as the
  subterm of a term in a reduction sequence starting from an initial program
  of the form $\letrecin D {\locsubs t \emptyset \emptyset}$.
\end{definition}

\subsection{A Sketch of Correctness (Termination)}

Due to space limitations, we moved our proofs of correctness and
completeness for this annotated reduction algorithm to
\appendixref{appendix:proofs}. It proves the following results.

\begin{lemma*}
  If $\bar t$ is reachable and reduces, in the annotated system, to
  a $\beta$-normal form $\bar v$, then $\floor {\bar v}$ is the
  $\beta$-normal form of $\floor {\bar t}$.
\end{lemma*}

\begin{theorem*}[Correctness]
  Annotated reduction is strongly normalizing: it either reduces to
  a $\beta$-normal form or reduces to a blocked redex in a finite
  number of reduction steps.
\end{theorem*}

\begin{theorem*}[Completeness]
  If an annotated program $\bar t$ contains a blocked redex, then its
  unannotated erasure $\floor {\bar t}$ admits an infinite reduction
  sequence.
\end{theorem*}

In this section we will merely sketch our termination argument.

The general approach is to use a termination measure. We first define
a measure, that is, a function from our terms into a well-ordered
set -- a set with an order relation, such that there do not exist
infinite strictly-decreasing sequences. Then we prove that the our
annotated reduction strictly decreases the measure of terms.

We can first define a measure on our traces $l$. There is a finite set
of type declarations in our system, and a trace can contain each type
constructor at most one, so there is a largest possible trace $L$ that
contains all type constructors. An application annotated with this
trace cannot be reduced. Any trace $l$ can then be measured by the
length difference $\mathsf{length}(L) - \mathsf{length}(l)$,
a natural number.

Then the question is how to extend this measure on traces into
a measure on annotated terms. One could think of using the
(measure of) the trace of the head application, but this does not
decrease during reduction if a subterm contains a larger trace and
ends up in head position. (Note: we state correctness for any
reduction strategy on annotated terms, not necessarily head reduction.)

A second idea is to measure a term $\bar t$ by the set of
(measures of) the traces that occur inside it -- technically the
multiset of traces, using the
\href{https://en.wikipedia.org/wiki/Dershowitz-Manna_ordering}{multiset
  ordering}. But this is not decreasing either: when we reduce an
application $\loc {f {\Fam i {\bar t_i}}} l$, we may duplicate its
arguments $t_i$ arbitrarily many times, and those may contain traces
that are strictly larger than $l$, resulting in a larger overall
measure for the reduced term.

The trick to make the proof work, which was suggested to us by Irène
Waldspurger, is to use multisets of multisets: we measure each subterm
in our annotated term by the \emph{path} from this subterm to the root
of the term, seen as a multiset of traces of applications. And then we
measure our term by the multiset of measures of its subterms.

For example, for the term
$\loc {f (\loc g {l_g}, \loc h {l_h})} {l_f}$, the measure of the root
subterm is $\dbraces {l_f}$ (we use double braces for multisets rather
than sets), the measure of the $g$ subterm is $\dbraces {l_f, l_g}$,
and the measure of the $h$ subterm is $\dbraces {l_f, l_h}$, so the
measure of the whole term is
$\dbraces {\dbraces {l_f}, \dbraces{l_f, l_g}, \dbraces{l_f, l_h}}$.

Consider now a subterm $\bar u$ of an argument of the redex
$\loc {f {\Fam i {\bar t_i}}} l$. The measure of its head application
may be larger than the measure of $l$, but its measure as a subterm is
the \emph{path} to the root, which contains $l$. When the application of
$f$ gets replaced by new subterms of a strictly smaller trace $f, l$,
then the path from $\bar u$ to the root will change, $f$ gets replaced
by these new nodes in the path, so the path measure of $\bar u$ decreases
strictly. This works even if this subterm $\bar u$ gets duplicated by
expansion: we get several copies, but at a strictly smaller measure,
so we are still decreasing for the multiset measure.

For example, if $\loc {f (\loc g {l_g}, \loc h {l_h})} {l_f}$ reduces
into $\loc {f' (\loc g {l_g}, \loc g {l_g})} {(l_f, f)}$, the subterm
$h$ has been erased and the subterm $g$ has been duplicated; its
measure is now $\dbraces {(l_f, f), l_g}$, which is strictly smaller
than its previous measure $\dbraces {l_f, l_g}$. The measure of the whole term changed from
$\dbraces {\dbraces {l_f}, \dbraces{l_f, l_g}, \dbraces{l_f, l_h}}$
to the strictly smaller
$\dbraces {\dbraces {(l_f, f)}, \dbraces{(l_f, f), l_g}, \dbraces{(l_f, f), l_g}}$.

\section{On \texttt{cpp}}
\label{sec:cpp}

Our on-the-fly termination checking algorithm is related to the
macro expansion algorithm of the \code{cpp} preprocessor for C, as
presented in the C11 standard. Quoting the standard:

\begin{quotation}
6.10.3.4 (2) If the name of the macro being replaced is found during this scan of the replacement list (not including the rest of the source file’s preprocessing tokens), it is not replaced. Furthermore, if any nested replacements encounter the name of the macro being replaced, it is not replaced. These nonreplaced macro name preprocessing tokens are no longer available for further replacement even if they are later (re)examined in contexts in which that macro name preprocessing token would otherwise have been replaced.
\end{quotation}

In this context, the ``replacement list'' denotes the body of
a function-like macro definition -- here we only consider function-like
macros, \code{#define FOO(..) ...} rather than \code{#define FOO
  ...}. The standard explains that a macro must not be expanded from
its own definition or from a ``nested replacement''. This corresponds
to our idea of blocking redexes whose function name occur in their own
trace. The idea that this non-replacement information remains active
``in later contexts'' corresponds to the idea of carrying annotations
around in subterms as the computation proceeds.

The phrasing of the standard is not very clear! In the 1980s, Dave
Prosser worked on a strategy to ensure that macro replacement always
terminates by disallowing dangerous cyclic/recursive macros, and wrote
a careful algorithm to allow as much replacement as would be possible
without -- hopefully -- endangering non-termination. This algorithm
was published as pseudo-code in a technical note
\citet*{prosser-86}. The C89 standard committee then translated Dave
Prosser's pseudo-code into the obscure prose that became the standard
text.

We know about David Prosser's pseudo-code today thanks to Diomidis
Spinellis who spent ``five years trying to implement a fully
conforming C preprocessor''; Spinellis wrote an annotated version of
Prosser's pseudo-code that explains the code: \citet*{spinellis-2008}.

Prosser's algorithm is strongly related to our termination-monitoring
algorithm -- it was not an inspiration as we were unfortunately
unaware of the connection when designing our own. We are not aware of
any proof of correctness for Prosser's algorithm -- that it does
guarantee termination.

In this section, we will compare the two algorithms. We will only
consider the ``core'' fragment of \code{cpp} macros, consisting of the
function-like macros that use their formal parameters directly -- no
conditionals, no use of stringization or concatenation, etc.

\subsection{First-Order and Closed-Higher-Order Function Macros}

The \code{cpp} preprocessor performs $\beta$-reduction, but also
parsing: it starts from a linear sequence of tokens instead of an
abstract syntax tree. It is difficult to reason at the level of
sequences of tokens, in particular about macros that generate
unbalanced sequences of parentheses; we will not attempt to do so
here. Let us only consider (core) macros whose terms are
well-parenthesized. What is their expressivity in terms of
a $\lambda$-calculus?

We define the \emph{first-order fragment} of core macros as the
fragment where all bound occurrences of a macro name \code{foo} are
syntactically an application \code{foo(...)} -- \code{foo} is
immediately followed by well-bracketed parentheses.

This is the fragment that the vast majority of C programmers use. But
it is possible to write (well-parenthesized) macros outside that
fragment, whose reduction behavior is less clear. We will mention two
examples, which we call the \code{NIL} example and the \code{a(a)}
example:

\begin{minipage}{0.6\linewidth}
\begin{lstlisting}[language=C]
#define NIL(xxx) xxx
#define G0(arg) NIL(G1)(arg)
#define G1(arg) NIL(arg)

G0(42) // $\rewto$ NIL(G1)(42) $\rewto$ G1(42)
        // $\rewto$ NIL(42) $\rewto$ 42
\end{lstlisting}
\end{minipage}
\begin{minipage}{0.4\linewidth}
\begin{lstlisting}[language=C]
#define a(x) b
#define b(x) x

a(a)(a)(a) // $\rewto$ b(a)(a)
            // $\rewto$ a(a) $\rewto$ b
\end{lstlisting}
\end{minipage}

In the \code{NIL} example, \code{G1} is used in non-applied position
in the definition of \code{G0}. In the \code{a(a)} example, the
second occurrence of \code{a} in \code{a(a)} is not in application
position.

We call this general case a \emph{closed-higher-order} language: it is
\emph{higher-order} in the sense that functions (macro names) can be
passed as parameters and returned as results, but those functions
remain \emph{closed}: functions cannot be declared locally and capture
lexical variables. This language also corresponds to simply-typed
\emph{supercombinators}~\citep*{turner-79,hughes-82}. (It is also
similar to the use of function \emph{pointers} in C, but we use macro
names instead of runtime addresses.)

We can consider typed or untyped versions of this closed-higher-order
language. The simply-typed version is a fragment of the simply-typed
$\lambda$-calculus with recursion, so its halting problem is
decidable. The untyped version is Turing-complete, just like most
extensions with more constructors or richer type systems, so their
halting problem is undecidable. C macro authors probably do not
consider typing their macros, so they work in the untyped version, but
an extension of ML with higher-kinded type definitions would
correspond to the typed version. Our algorithm does not depend on
types, so it can work on either version.

In the rest of this section we will detail the following claims:

\vspace{-.5em}
\begin{fact} Our algorithm gives the same result as Dave Prosser's
  on the first-order macro fragment.
\end{fact}
\vspace{-1.5em}
\begin{fact} Our algorithm extends to the closed-higher-order
  macro fragment, and remains correct.
\end{fact}
\vspace{-1.5em}
\begin{fact} Neither our algorithm nor Prosser's are complete on the
  closed-higher-order macro fragment.
\end{fact}

\subsection{The C and C++ Standard Bodies on Closed-Higher-Order Macros}

Those two examples, \code{NIL} and \code{a(a)}, come from the Defect
Reports 017 of the C and C++ standard committee~\citep*{dr017}. Since
the C89 standard was published, programmers have asked for
clarifications about the replacement behavior that would dictate how
those examples should behave -- we have indicated possible reduction
sequences in comments, but implementations in the wild would behave
differently and often stop before reaching the normal form.

In the first few years, the C standard committee refused to provide
clarifications -- we understand that those examples were perceived as
unrealistic and absent from real-world C code. See the answers to
questions 17 and 23 in the Defect report 017~\citep*{dr017}. It later
became clear that some C or C++ programmers made real-world use of the
non-first-order fragment, and lately standard bodies have moved
towards actually specifying this behavior, choosing to honor Dave
Prosser's original intent to reduce as much as possible. See in
particular the discussion of the \code{NIL} example above in the
document N3882 from the C++ standard body~\citep*{n3882}.

\subsection{Relating our Algorithm to Dave Prosser's Pseudo-Code}
\label{sec:relating-to-dave-prosser}

In \appendixref{appendix:cpp}, we show and explain Dave Prosser's
algorithm, we relate it to our own termination-monitoring algorithm
(they are different but related), we claim that they provide the same
expansions in the first-order case, and finally we explain how
Dave Prosser's algorithm works outside the first-order fragment.

We are not aware of any previous proof that Dave Prosser's algorithm
terminates on all inputs; our comparison provides a proof in the
first-order case for the core fragment, but the closed-higher-order
case remains open.

\subsection{Adapting our Algorithm to Closed-Higher-Order Macros}
\label{subsec:closed-higher-order-calculus}

For the purposes of OCaml head shape computation, we have only
formulated our termination-monitoring algorithm on the first-order
$\lambda$-calculus with recursion. We now extend our first-order calculus
from Section~\ref{subsec:first-order-calculus} to the
closed-higher-order fragment, and show that our termination-monitoring
algorithm still enforces termination. In the ML world, this would
correspond to extending ML declarations to higher-kinded type
parameters -- a feature present in Haskell.
\begin{mathpar}
  \begin{array}{lcl}
    p & \bnfeq & \letrecin D t \\
    D & \bnfeq & \emptyset \bnfor D, f {\Fam i {x_i}} = t \\
    f & & \text{function name} \\
    x & & \text{first-order variable} \\
    t & \bnfeq & x \bnfor \mathhl{f} \bnfor \mathhl{t} {\Fam i {t_i}} \\
  \end{array}

  \begin{array}{lcl}
    \bar p & \bnfeq & \letrecin D {\bar t}                    \\
    \bar t & \bnfeq & x \bnfor \loc {\mathhl{t} {\Fam i {\bar t_i}}} l \\
    l & \bnfeq & \emptyset \bnfor l, f \quad \text{($f \notin l$)} \\
    C[\hole] & \bnfeq & \hole \bnfor \loc {\mathhl{t} {(\Fam i {\bar t_i}, C[\hole], \Fam j {\bar t_j}})} l \\
  \end{array}
\end{mathpar}

We have highlighted above the changes compared to the first-order
case: we add function names $f$ as first-class terms, and relax our
syntax of application from an application of a known function
$\loc {f \Fam i {t_i}} l$ to the application of an arbitrary term
$\loc {t \Fam i {t_i}} l$. Note that the annotation $l$ here is the
trace of the whole application node, not the trace of the function
symbol -- function symbols do not carry an annotation.\footnote{We
  conjecture that our treatment corresponds to the intersection of the
  hide sets of the function name and closing parenthesis in Dave
  Prosser's algorithm. The function name and closing parenthesis may
  have been generated by different macro expansions, and the
  intersection computes a nearest common ancestor. Our application
  nodes make this nearest common ancestor explicit in the AST.}

The definition of reduction, repeated below for reference, is
unchanged. But it now implicitly restricts redexes to the case where
the left-hand-side of the function application has first been reduced
to a function name $f$.
\begin{mathpar}
\infer
{(f {\Fam i {x_i}} = t') \in D \\ f \notin l}
{
  \letrecin D C[\loc {f {\Fam i {\bar t_i}}} l]
  \rewto
  \letrecin D C[\locsubs {t'} {\Fam i {x_i \leftarrow \bar t_i}} {l,f}]
}
\end{mathpar}

This extension can trivially express the closed-higher-order examples
that we discussed so far, as well as the famous non-terminating
$\lambda$-term $(\lam x {x~x}) ~ (\lam x {x~x})$.

\begin{minipage}{0.4\linewidth}
\begin{lstlisting}
let rec nil(x) = x
    and g0(arg) = nil(g1)(arg)
    and g1(arg) = nil(arg)
in g0(fortytwo)
\end{lstlisting}
\end{minipage}%
\begin{minipage}{0.25\linewidth}
\begin{lstlisting}
let rec a(x) = b
    and b(x) = x
in a(a)(a)(a)
\end{lstlisting}
\end{minipage}%
\begin{minipage}{0.35\linewidth}
\begin{lstlisting}
let rec delta(x) = x(x)
in delta(delta)
\end{lstlisting}
\end{minipage}

The \code{NIL} and \code{a(a)} examples reduce without getting
stuck. The \code{delta(delta)} example gets stuck without looping.

\begin{minipage}{0.4\linewidth}
\begin{lstlisting}
   g0(fortytwo)[]
$\rewto$ nil(g1)[g0](fortytwo)[g0]
$\rewto$ g1(fortytwo)[g0]
$\rewto$ nil(fortytwo)[g0]
$\rewto$ fortywto
\end{lstlisting}
\end{minipage}%
\begin{minipage}{0.25\linewidth}
\begin{lstlisting}
   a(a)[](a)[](a)[]
$\rewto$ b(a)[](a)[]
$\rewto$ a(a)[]
$\rewto$ b
\end{lstlisting}
\end{minipage}%
\begin{minipage}{0.35\linewidth}
\begin{lstlisting}
   delta(delta)[]
$\rewto$ delta(delta)[delta]
\end{lstlisting}
\end{minipage}

\begin{theorem}
  Our termination-monitoring algorithm remains correct for this
  closed-higher-order language: the annotated language is strongly
  normalizing.
\end{theorem}

The proof of this result reuses the technical machinery of the
correctness proof for the first-order fragment. It is detailed in
\appendixref{app:closed-higher-order-correctness-proof}.

\subsection{Both Algorithms are Incomplete Outside the First-Order Fragment}

Here is a counterexample to completeness in the closed-higher-order
fragment:
\begin{lstlisting}[language=C]
#define f(p,q) p(f(q,q))
#define id(x) x
#define stop(x) done
f(id,stop) // $\rewto$ id(f(stop,stop)) $\rewto$ f(stop,stop) 
            // $\rewto$ stop(f(stop,stop)) $\rewto$ done
\end{lstlisting}

This example requires two nested expansions of \code{f} to reduce to
its normal form \code{done}. Neither our algorithm nor Prosser's can
reach this normal form; they block at \code{f(stop,stop)}. Note that
this term is an example of a term that is only weakly normalizing: it
has a normal form, but also an infinite reduction sequence
\begin{lstlisting}
f(stop,stop) $\rewto$ stop(f(stop,stop)) $\rewto$ stop(stop(f(stop,stop))) $\rewto$ ...
\end{lstlisting}

\section{Related Work}
\label{sec:related-work}

There is a lot of work on value representations in programming
languages, including questions of how to optimize data
representation. We are not aware, however, of previous academic work
on custom (datatype-specific) representation of sum/coproduct types,
which requires checking that a candidate representation ensures
disjointedness.

We will discuss the state of representation optimization in
neighboring languages with sum types. In particular, the
``niche-filling'' optimizations of Rust are the most closely related
work, as they involve a form of disjointedness. To our knowledge they
have not previously been discussed in an academic context, except for
a short abstract~\cite{filling-a-niche}.

\subsection{Functional Programming Languages}

\paragraph{Haskell}

Haskell offers \code{newtype} for single-constructor unboxing. GHC
supports unboxed sum types among its unboxed value types. Unboxed
value types live in kinds different from the usual kind \code{Type} of
types whose value representation is uniform. GHC also supports an
UNPACK pragma~\citep*{GHC-UNPACK} on constructor arguments to require
that this argument be stored unboxed -- generalizing OCaml's unboxing
of float arrays and records. Haskell would still benefit from
constructor unboxing. Note that lifted types (containing lazy thunks)
would conflict with each other, limiting applicability -- one has to
use explicitly unlifted types, or Strict Haskell, etc.

\paragraph{MLton}

MLton can eliminate some boxing due to aggressive
specialization~\citep*{mlton}; for example, \code{(int * int) array}
will implicitly unbox the \code{(int * int)} tuple. Its relevant
optimizations are
\href{http://www.mlton.org/SimplifyTypes}{SimplifyTypes}, which
performs unboxing for datatypes with a single constructor
(after erasing constructors with uninhabited arguments) and
\href{http://www.mlton.org/DeepFlatten}{DeepFlatten},
\href{http://www.mlton.org/RefFlatten}{RefFlatten} which optimize
combinations of product types and mutable fields. The representation
of sum types with several (inhabited) constructors remains uniform.

\paragraph{Scala}

In Scala, representation questions are constrained by the JVM but also
by the high degree of dynamic introspectability expected. Even the
question of single-constructor unboxing is delicate. A widespread
\code{AnyVal} pattern has disappointing
performance~\citep*{the-high-cost-of-anyval-classes,scala-pre-SIP-unboxed-wrapper-types}
\, and dotty introduced a specific \emph{opaque type synonym}
feature~\citep*{scala-SIP-opaque-types,dotty-PR-opaque-types} to work
around this.

\paragraph{Specialization and representation optimizations}

Both MLton and Rust create opportunities for datatype representation
optimizations by performing aggressive monomorphization. In OCaml,
statically specializing the representation of \code{(int * int) array}
(as MLton does) or \code{Option<Box<T>>} (as Rust does) to be more
compact would not be possible, as polymorphic functions working with
\code{'a array} or \code{'a option} inputs are compiled once for
all instances of these datatypes.

On the other hand, a non-monomorphizing language could perfectly
implement and support representation optimizations for types that are
only used in a specialized context. For example, at the cost of some
code duplication, high-performance code could define its own %
\code{'a option_ref} code that looks similar to \code{('a option) ref}
but with a more compact representation. (In fact people do this today
using foreign or unsafe code; \code{'a ref_option} can be safely
expressed with inline records.) Finally, in
Section~\ref{subsec:shape-constraints-on-type-variables} we described
parameterized type definitions that support a quantification on the
head shape of the type parameter, which could provide a compromise
between genericity and unboxing opportunities.

We thus consider that monomorphization is a separate concern from
representation optimizations such as unboxing. Monomorphization
creates more opportunities for representation optimizations, with
well-known tradeoff in terms of compilation time and code
size. Languages could consider representation optimizations whether or
not they perform aggressive specialization.

\paragraph{Unboxed sums and active/view patterns}

Some languages, such as GHC Haskell and \Fsharp{}, support both
unboxed sums as ``value
types''~\citep*{ghc-unboxed-sums,fsharp-PR-struct-discriminated-unions}
and active/view
patterns~\citep*{fsharp-active-patterns,ghc-view-patterns} to apply
user-defined conversions during pattern matching. Combining those two
features can get us close to unboxed constructors in some contexts,
including our Zarith/gmp example. The idea is to have a very compact
primitive/opaque value representation, and expose a ``view'' of these
values in terms of unboxed sums -- in \Fsharp{} they are called
``struct-based discriminated unions''.

This style can combine a compact in-memory representation, yet provide
the usual convenience of pattern-matching, without extra
allocations -- even when crossing an FFI boundary. The overhead would
typically be higher than native constructor unboxing, but only by
a constant factor.

This approach is very flexible, it can be used to perform
representation optimizations that are not covered by
single-constructor unboxing alone. For example, it can be used to view
native integers into an unboxed sum type of
positive-or-negative-or-zero numbers. On the other hand, constructor
unboxing extends the ranges of memory layouts that can be expressed
directly in the language, as in our Zarith example; to use the
``view'' in this situation, we have to implement the representation
datatype and the view function in unsafe foreign code.

\subsection{Rust: Niche-Filling}
\label{subsec:rust-niche-filling}

Rust performs a form of constructor unboxing by applying so-called
``niche-filling rules''. The paradigmatic example of niche-filling is
unboxing \code{Option<A>} for all types \code{A} whose representation
contains a ``niche'' leaving one value available -- typically, if
\code{A} is a type of non-null pointers, the \code{0} word is never
a valid \code{A} representation and can be used to represent the
option's \code{None}.

Note that niche-filling in Rust is a qualitatively different
optimization from constructor unboxing in OCaml, as it does not remove
indirection through pointers. In Rust datatypes, sums (enums) are
unboxed and pointer indirections are fully explicit through types such
as \code{Box<T>}, and to our knowledge they are never optimized
implicitly. Niche-filling shrinks the width of (value/unboxed) sum
types by removing their tag byte; this is important when using
(unboxed) arrays of enums, and because all values of an enum are
padded to the width of its widest constructor. One could say that it
is an \emph{untagging} rather than \emph{unboxing} optimization -- but
a tag is a sort of box.

The current form of niche-filling was discussed in
\citet*{rust-RFC-1230} and implemented in two steps in
\citet*{rust-PR-45225} and \citet*{rust-PR-94075}. We understand that
it works as follows: among the constructors whose arguments have
maximal total width, try to find one with a ``niche'', a set of
impossible values, that is large enough to store the tag of the
sum. Then check that the arguments of all other constructors can be
placed in the remaining space -- outside the niche.

We can formulate niche-filling in our framework. Head shapes are
abstractions of sets of binary words -- all of the same size. In
particular, the complement of the head shape of a type is its set
of niches. Rust allows unboxing exactly when \emph{all} constructors
of a type can be ``unboxed'' in our sense -- they can be placed in
non-overlapping subsets of binary words. Rust will implicitly change
the placement of some constructor arguments to obtain more
non-overlapping cases, which is an instance of the ``unboxing by
transformation'' approach that we discussed in
Section~\ref{subsec:unboxing-by-transformation}.

Niche-filling has some limitations:
\begin{enumerate}
\item Filling a niche with a type B is a direct inclusion of the tag
  and values of B, possibly placed at non-zero offsets within A; it
  is not possible to transform the representation of B along the way
  (for example, if B has two values, it is not possible to use the
  niches of two separate non-null pointer fields, there must be a full
  padding bit available somewhere.)

  This limitation is shared with our work. (We point it out because
  some people expect magic from Rust optimizations.) One reason is that
  this guarantees that projecting the B out of the value is cheap,
  instead of requiring a possibly-expensive transformation.

\item Today niche-filling is implementation-defined in Rust, purely
  a compiler heuristic, and there is no interface for the user to
  control niche-filling behavior. There seem to be some rare cases
  where niche-filling in fact degrades performance by making
  pattern-matching slightly more complex; the only recourse for
  programmers is to use the \code{enum(C)} attribute to ask for an
  ABI-specified layout that disables all representation
  optimizations. (There is an instance of this workaround in the PR
  that generalized niche-filling,
  \href{https://github.com/rust-lang/rust/pull/94075}{94075}, out of
  concern for a potential performance regression in the Rust compiler
  itself.)

\item Rust currently only allows niche-filling when all constructors
  can be unboxed at once. This is less expressive than our approach
  where some constructors can be unboxed selectively while others are
  not. In the context of Rust, removing the tag bytes from \emph{all}
  constructors is the very point of niche-filling, so unboxing
  (untagging, really) only certain constructors selectively is not
  interesting. This suggests however that niche-filling as currently
  conceived in the Rust community is a specialized form of constructor
  unboxing, that is not necessarily suitable for other programming
  languages. Our presentation is more general.
\end{enumerate}

\subsection{Declarative Languages for Custom Datatype Representations}

\paragraph{Hobbit}

The experimental language Hobbit offered syntax to define algebraic
``bit-datatypes'' together with a mapping into a lower-level bit-level
representation for FFI purposes, with support for constructing and
matching the lower-level value directly. See in particular
\citet{hobbit-bitdata}, Section ``2.2.4 Junk and Confusion?'' on page
9, which refers to the coproduct property as a ``no junk and
confusion'' guarantee, and mentions that their design in fact allows
confusion: they are not forbidden by a static analysis, but a warning
detects them in some cases.

One expressivity limitation is that ``bit-datatypes'' form a ``closed
universe'': all types mentioned in a ``bit-datatype'' declaration must
themselves be bit-data whose representation is available to the
compiler, it does not seem possible to specify only a part of the
variant value as bit-data, and have arbitrary values of the language
(that may not have a bit-data encoding) for other arguments.

\paragraph{Cogent / Dargent}

Dargent~\citep*{dargent} is a data layout language for Cogent, focused
on providing flexible data representation choices for
structure/product types. Users can define custom layout rules for
record types, the system generates custom getters and setters for each
field and verifies that each pair of a getter and setter satisfies the
expected specification of being a lens. Cogent also has sum types, but
Dargent does not allow unboxed layouts; it requires a bitset to store
a unique tag for each case. Extending this formalization to richer
representations of sum types would allow constructor unboxing. This
could be done by generating a pattern-matching function and a family
of constructor functions, and checking that those satisfy the expected
properties of a coproduct -- disjointedness.

(There are many other data-description languages (DDLs), but they tend
to not offer specific support for sum types / disjoint unions,
especially not approaches to unbox them. The Dargent Related Work
mentions the prominent previous works in this area.)

\paragraph{\texttt{Ribbit}: Bit Stealing Made Legal}

Simultaneous to our work, \citet*{bit-stealing-made-legal} propose
a new domain-specific language to express data layouts. It is designed
for efficiency rather than for interoperability with foreign data,
supports flexible representation optimizations for sum types, and
formalizes (on paper) the ``no-confusion'' correctness condition for
sums.

Having a full layout DSL is more expressive than our approach. The
authors had both OCaml and Rust in mind when elaborating their
proposed DSL, called \texttt{Ribbit}, which is thus expressive enough to
express constructor unboxing (Zarith is mentioned as a use-case), as
well as Rust niche-filling optimizations. The language currently
suffers from the same ``closed universe'' limitation as Hobbit, and
needs to be extended to deal with abstract types or non-specialized
type parameters, but those are not fundamental limitations of the
approach.

It would be nice to restructure our work as elaborating the
user-facing feature into an intermediate representation of data layouts,
and \textsc{Ribbit} would be a good candidate for this.

The authors discuss efficient compilation of pattern-matching
involving datatypes with a custom layout; they moved from matrix-based
pattern-matching compilation (as used in the OCaml compiler) to
decision-tree-based compilation to handle more advanced, ``irregular''
representations. (The \textsc{Ribbit} prototype does not yet support
real programs, so it is hard to draw performance conclusions; in
contrast our prototype supports the full OCaml language.) So far our
intuition is that constructor unboxing, and the various extension we
considered, are regular enough that matrix-based compilation can give
good results and no invasive changes are required in the OCaml
compiler. But this question deserves further study.

\subsection{Deciding Termination of Recursive Rewrite Rules}

Our termination-monitoring algorithm provides an alternative proof of
decidability of the halting problem for the first-order recursive
calculus. This result was well-known (inside a different scientific
sub-community than ours), but we want to point out that we provide
a new (to our knowledge) decision \emph{algorithm}, and that having
a simple yet efficient algorithm is important to our application. By
``termination-monitoring'' we mean that our algorithm only adds some
bookkeeping information to computations that we have to perform in any
case to compute a normal form.

The simplest proof of decidability of termination for the first-order
recursive calculus that we have found (thanks to Pablo Barenbaum) is
\citet*{first-order-halting-problem}. Seen as a decision algorithm,
this proof is terrible: it performs an exponential computation of all
normal forms of all types definitions.

On the other hand, other termination proofs exist that also scale to
the higher-order cases, and rely on evaluating the program in
a semantic domain where base types are interpreted as the two-valued
Scott domain $\{\bot, \top\}$, where $\bot$ represents
potentially-non-terminating terms and $\top$ represent terminating
terms. Terms at more complex types are interpreted as more complex
domains, but (in absence of recursive types) they remain finite, so in
particular recursive definitions, interpreted as fixpoints, can be
computed by finite iteration. This approach in particular underlies
the termination proofs in
\citet*{higher-order-halting-problem,lmcs:1567}. This argument has
a computational feel, but naively computing a fixpoint for each
recursive definition in the type environment is still more work than
we want to perform. It is possible that a clever on-demand computation
strategy for those fixpoints would correspond to our algorithm.

\newpage

\begin{acks}
Jeremy Yallop proposed the general idea of this feature in March
2020~\citep*{unboxing-RFC}. There are various differences between our
implementation and Jeremy Yallop's initial proposal, mostly in the
direction of simplifying the feature by keeping independent aspects
for later, and our formal developments were a substantial new
effort. Jeremy Yallop also suggested Zarith as a promising application
of the feature.

Nicolas Chataing and Gabriel Scherer worked on this topic together in
spring-summer 2021, during a master internship of Nicolas Chataing
supervised by Gabriel Scherer. Most of the the ideas presented in
Sections 3, 4, 5 were obtained during this period, along with the
termination-checking algorithm, but without a convincing proof of
termination. An important part of the work not detailed in this
article was the implementation of a prototype within the OCaml
compiler, which required solving difficult software engineering
questions about dependency cycles between various processes in the
compiler (checking type declarations, managing the typing environment,
computing type properties for optimization, accessing type properties
during compilation). Around and after the end of the internship,
Gabriel Scherer implemented the Zarith case study, implemented small
extensions, and set out to prove termination -- soundness of the
on-the-fly termination checking algorithm. Gabriel Scherer also did
the writing, and intends to work on upstreaming the feature.

Stephen Dolan entered the picture with high-quality remarks on the
work, in particular identifying the notion of benign cycles, and
remarking that the termination algorithm is similar to the \code{cpp}
specification, with precise pointers to the work of Dave Prosser via
Diomidis Spinellis that Gabriel Scherer used to propose a detailed
comparison in Section 7.

The most important contribution of Stephen Dolan is probably that he
killed the first four or five attempts at a termination proof.  Once
we collectively got frustrated with Stephen Dolan's repeated ability
to strike down proposed termination arguments, Irène Waldspurger played
the key role of suggesting a correct argument: one has to use
multisets of multisets of nodes, instead of trying to make do with
multisets of nodes.

Jacques-Henri Jourdan remarked on a relation to the Coq
value representation optimization in \code{native_compute}, Guillaume
Melquiond provided extra information and Guillaume Munch-Maccagnoni
tried some early experiments. We have yet to investigate this idea in
full, but it may provide a nice simplification of this low-level
aspect of \code{native_compute}.

Adrien Guatto greatly improved the presentation of a previous version
of this work by a healthy volume of constructive criticism. Anonymous
POPL reviewers also provided excellent feedback and engaging
questions, in particular Section~\ref{subsec:unboxing-existing-tricks}
is due to their curiosity.

Finally, the \texttt{types} mailing-list has helped tremendously in
finding previous work on decidability of recursive rewrite systems
(\href{http://lists.seas.upenn.edu/pipermail/types-list/2023/002488.html}{thread}).
\end{acks}


\bibliography{unboxing}


\begin{thebibliography}{32}


\ifx \showCODEN    \undefined \def \showCODEN     #1{\unskip}     \fi
\ifx \showDOI      \undefined \def \showDOI       #1{#1}\fi
\ifx \showISBNx    \undefined \def \showISBNx     #1{\unskip}     \fi
\ifx \showISBNxiii \undefined \def \showISBNxiii  #1{\unskip}     \fi
\ifx \showISSN     \undefined \def \showISSN      #1{\unskip}     \fi
\ifx \showLCCN     \undefined \def \showLCCN      #1{\unskip}     \fi
\ifx \shownote     \undefined \def \shownote      #1{#1}          \fi
\ifx \showarticletitle \undefined \def \showarticletitle #1{#1}   \fi
\ifx \showURL      \undefined \def \showURL       {\relax}        \fi
\providecommand\bibfield[2]{#2}
\providecommand\bibinfo[2]{#2}
\providecommand\natexlab[1]{#1}
\providecommand\showeprint[2][]{arXiv:#2}

\bibitem[A\u{g}acan(2016)]%
        {ghc-unboxed-sums}
\bibfield{author}{\bibinfo{person}{{\"O}mer~S{\'i}nan A\u{g}acan}.}
  \bibinfo{year}{2016}\natexlab{}.
\newblock \bibinfo{title}{GHC unboxed sums}.
\newblock
\newblock
\urldef\tempurl%
\url{https://github.com/ghc/ghc/commit/714bebff44076061d0a719c4eda2cfd213b7ac3d}
\showURL{%
\tempurl}


\bibitem[Bartell-Mangel(2022)]%
        {filling-a-niche}
\bibfield{author}{\bibinfo{person}{Noah~Lev Bartell-Mangel}.}
  \bibinfo{year}{2022}\natexlab{}.
\newblock \bibinfo{title}{{Filling a Niche: Using Spare Bits to Optimize Data
  Representations}}.
\newblock
\newblock
\urldef\tempurl%
\url{https://www.noahlev.org/papers/popl22src-filling-a-niche.pdf}
\showURL{%
\tempurl}
\newblock
\shownote{POPL'22 student research presentation}.


\bibitem[Baudon, Radanne and Gonnord(2023)]%
        {bit-stealing-made-legal}
\bibfield{author}{\bibinfo{person}{Tha{\"i}s Baudon}, \bibinfo{person}{Gabriel
  Radanne}, {and} \bibinfo{person}{Laure Gonnord}.}
  \bibinfo{year}{2023}\natexlab{}.
\newblock \showarticletitle{{Bit-Stealing Made Legal}}. In
  \bibinfo{booktitle}{\emph{{ICFP}}}.
\newblock
\urldef\tempurl%
\url{https://doi.org/10.1145/3607858}
\showDOI{\tempurl}


\bibitem[Beingessner(2015)]%
        {rust-RFC-1230}
\bibfield{author}{\bibinfo{person}{Aria Beingessner}.}
  \bibinfo{year}{2015}\natexlab{}.
\newblock \bibinfo{title}{Rust RFC 1230: More Exotic Enum Layout
  Optimizations}.
\newblock
\newblock
\urldef\tempurl%
\url{https://github.com/rust-lang/rfcs/issues/1230}
\showURL{%
\tempurl}


\bibitem[Benfield(2022)]%
        {rust-PR-94075}
\bibfield{author}{\bibinfo{person}{Michael Benfield}.}
  \bibinfo{year}{2022}\natexlab{}.
\newblock \bibinfo{title}{rustc PR 94075: Use niche-filling optimization even
  when multiple variants have data}.
\newblock
\newblock
\urldef\tempurl%
\url{https://github.com/rust-lang/rust/pull/94075}
\showURL{%
\tempurl}


\bibitem[Boespflug, Dénès and Grégoire(2011)]%
        {native-compute}
\bibfield{author}{\bibinfo{person}{Mathieu Boespflug}, \bibinfo{person}{Maxime
  Dénès}, {and} \bibinfo{person}{Benjamin Grégoire}.}
  \bibinfo{year}{2011}\natexlab{}.
\newblock \showarticletitle{Full Reduction at Full Throttle}. In
  \bibinfo{booktitle}{\emph{CPP}}.
\newblock
\urldef\tempurl%
\url{https://inria.hal.science/hal-00650940}
\showURL{%
\tempurl}


\bibitem[Burtescu(2017)]%
        {rust-PR-45225}
\bibfield{author}{\bibinfo{person}{Eduard-Mihai Burtescu}.}
  \bibinfo{year}{2017}\natexlab{}.
\newblock \bibinfo{title}{rustc PR 45225: Refactor type memory layouts and
  ABIs, to be more general and easier to optimize}.
\newblock
\newblock
\urldef\tempurl%
\url{https://github.com/rust-lang/rust/pull/45225}
\showURL{%
\tempurl}


\bibitem[Chan(2017)]%
        {scala-pre-SIP-unboxed-wrapper-types}
\bibfield{author}{\bibinfo{person}{Lloyd Chan}.}
  \bibinfo{year}{2017}\natexlab{}.
\newblock \bibinfo{title}{Scala Pre-SIP: Unboxed wrapper types}.
\newblock
\newblock
\urldef\tempurl%
\url{https://contributors.scala-lang.org/t/pre-sip-unboxed-wrapper-types/987}
\showURL{%
\tempurl}


\bibitem[Chen, Lafont, O'Connor, Keller, McLaughlin, Jackson and
  Rizkallah(2023)]%
        {dargent}
\bibfield{author}{\bibinfo{person}{Zilin Chen}, \bibinfo{person}{Ambroise
  Lafont}, \bibinfo{person}{Liam O'Connor}, \bibinfo{person}{Gabriele Keller},
  \bibinfo{person}{Craig McLaughlin}, \bibinfo{person}{Vincent Jackson}, {and}
  \bibinfo{person}{Christine Rizkallah}.} \bibinfo{year}{2023}\natexlab{}.
\newblock \showarticletitle{Dargent: A Silver Bullet for Verified Data Layout
  Refinement}.
\newblock \bibinfo{journal}{\emph{PACMPL}} \bibinfo{volume}{7},
  \bibinfo{number}{POPL}, Article \bibinfo{articleno}{47} (\bibinfo{date}{Jan}
  \bibinfo{year}{2023}), \bibinfo{numpages}{27}~pages.
\newblock
\urldef\tempurl%
\url{https://doi.org/10.1145/3571240}
\showDOI{\tempurl}


\bibitem[Colin, Lepigre and Scherer(2019)]%
        {mutual-unboxing}
\bibfield{author}{\bibinfo{person}{Simon Colin}, \bibinfo{person}{Rodolphe
  Lepigre}, {and} \bibinfo{person}{Gabriel Scherer}.}
  \bibinfo{year}{2019}\natexlab{}.
\newblock \showarticletitle{{Unboxing Mutually Recursive Type Definitions in
  OCaml}}. In \bibinfo{booktitle}{\emph{{JFLA 2019}}}.
\newblock
\urldef\tempurl%
\url{https://hal.inria.fr/hal-01929508}
\showURL{%
\tempurl}


\bibitem[Compall(2017)]%
        {the-high-cost-of-anyval-classes}
\bibfield{author}{\bibinfo{person}{Stephen Compall}.}
  \bibinfo{year}{2017}\natexlab{}.
\newblock \bibinfo{title}{Blog post: the high cost of AnyVal classes}.
\newblock
\newblock
\urldef\tempurl%
\url{https://failex.blogspot.com/2017/04/the-high-cost-of-anyval-subclasses.html}
\showURL{%
\tempurl}


\bibitem[Diatchki, Jones and Leslie(2005)]%
        {hobbit-bitdata}
\bibfield{author}{\bibinfo{person}{Iavor~S. Diatchki}, \bibinfo{person}{Mark~P.
  Jones}, {and} \bibinfo{person}{Rebekah Leslie}.}
  \bibinfo{year}{2005}\natexlab{}.
\newblock \showarticletitle{High-Level Views on Low-Level Representations}. In
  \bibinfo{booktitle}{\emph{ICFP'05}}.
\newblock
\urldef\tempurl%
\url{http://web.cecs.pdx.edu/~mpj/pubs/bitdata-icfp05.pdf}
\showURL{%
\tempurl}


\bibitem[Granlund and contributors(1991)]%
        {gmp}
\bibfield{author}{\bibinfo{person}{Torbjörn Granlund} {and}
  \bibinfo{person}{contributors}.} \bibinfo{year}{1991}\natexlab{}.
\newblock \bibinfo{title}{GMP}.
\newblock
\newblock
\urldef\tempurl%
\url{https://gmplib.org/}
\showURL{%
\tempurl}


\bibitem[Hughes(1982)]%
        {hughes-82}
\bibfield{author}{\bibinfo{person}{John Hughes}.}
  \bibinfo{year}{1982}\natexlab{}.
\newblock \showarticletitle{Super-Combinators a New Implementation Method for
  Applicative Languages}. In \bibinfo{booktitle}{\emph{Proceedings of the 1982
  ACM Symposium on LISP and Functional Programming (LFP)}}.
\newblock
\urldef\tempurl%
\url{https://doi.org/10.1145/800068.802129}
\showURL{%
\tempurl}


\bibitem[Khasidashvil(2020)]%
        {first-order-halting-problem}
\bibfield{author}{\bibinfo{person}{Zurab Khasidashvil}.}
  \bibinfo{year}{2020}\natexlab{}.
\newblock \showarticletitle{{A short proof of the decidability of normalization
  in recursive program schemes}}. In \bibinfo{booktitle}{\emph{Shalva
  Pkhakadze's Festschrift, AMIM Vol. 25 No. 2}}.
\newblock
\urldef\tempurl%
\url{http://www.viam.science.tsu.ge/Ami/2020_2/5_zura.pdf}
\showURL{%
\tempurl}


\bibitem[Marlow(2003)]%
        {GHC-UNPACK}
\bibfield{author}{\bibinfo{person}{Simon Marlow}.}
  \bibinfo{year}{2003}\natexlab{}.
\newblock \bibinfo{title}{GHC's UNPACK pragma}.
\newblock
\newblock
\urldef\tempurl%
\url{https://github.com/ghc/ghc/commit/abbc5a0be1df84a33015470319062ed7a3aa3153}
\showURL{%
\tempurl}


\bibitem[Min{\'e} and Leroy(2012)]%
        {zarith}
\bibfield{author}{\bibinfo{person}{Antoine Min{\'e}} {and}
  \bibinfo{person}{Xavier Leroy}.} \bibinfo{year}{2012}\natexlab{}.
\newblock \bibinfo{title}{Zarith}.
\newblock
\newblock
\urldef\tempurl%
\url{https://github.com/ocaml/Zarith/}
\showURL{%
\tempurl}


\bibitem[Odersky and Moors(2018)]%
        {dotty-PR-opaque-types}
\bibfield{author}{\bibinfo{person}{Martin Odersky} {and}
  \bibinfo{person}{Adriaan Moors}.} \bibinfo{year}{2018}\natexlab{}.
\newblock \bibinfo{title}{dotty PR 5300: Opaque types}.
\newblock
\newblock
\urldef\tempurl%
\url{https://github.com/lampepfl/dotty/pull/5300}
\showURL{%
\tempurl}


\bibitem[Osheim, Cantero and Doeraene(2017)]%
        {scala-SIP-opaque-types}
\bibfield{author}{\bibinfo{person}{Erik Osheim}, \bibinfo{person}{Jorge~Vicente
  Cantero}, {and} \bibinfo{person}{Sébastien Doeraene}.}
  \bibinfo{year}{2017}\natexlab{}.
\newblock \bibinfo{title}{Scala SIP 35: Opaque types}.
\newblock
\newblock
\urldef\tempurl%
\url{https://contributors.scala-lang.org/t/pre-sip-unboxed-wrapper-types/987}
\showURL{%
\tempurl}


\bibitem[Peyton-Jones(2007)]%
        {ghc-view-patterns}
\bibfield{author}{\bibinfo{person}{Simon Peyton-Jones}.}
  \bibinfo{year}{2007}\natexlab{}.
\newblock \bibinfo{title}{GHC view patterns}.
\newblock
\newblock
\urldef\tempurl%
\url{https://gitlab.haskell.org/ghc/ghc/-/wikis/view-patterns}
\showURL{%
\tempurl}


\bibitem[Plotkin(2022)]%
        {higher-order-halting-problem}
\bibfield{author}{\bibinfo{person}{Gordon Plotkin}.}
  \bibinfo{year}{2022}\natexlab{}.
\newblock \bibinfo{title}{{Recursion does not always help}}.
  (\bibinfo{year}{2022}).
\newblock
\urldef\tempurl%
\url{https://arxiv.org/pdf/2206.08413.pdf}
\showURL{%
\tempurl}


\bibitem[Prosser(1986)]%
        {prosser-86}
\bibfield{author}{\bibinfo{person}{Dave Prosser}.}
  \bibinfo{year}{1986}\natexlab{}.
\newblock \bibinfo{title}{{X3J11/86-196}: Complete macro expansion algorithm}.
\newblock
\newblock
\urldef\tempurl%
\url{https://www.spinellis.gr/blog/20060626/x3J11-86-196.pdf}
\showURL{%
\tempurl}


\bibitem[Salvati and Walukiewicz(2015)]%
        {lmcs:1567}
\bibfield{author}{\bibinfo{person}{Sylvain Salvati} {and} \bibinfo{person}{Igor
  Walukiewicz}.} \bibinfo{year}{2015}\natexlab{}.
\newblock \showarticletitle{{Using models to model-check recursive schemes}}.
\newblock \bibinfo{journal}{\emph{{Logical Methods in Computer Science}}}
  \bibinfo{volume}{{Volume 11, Issue 2}} (\bibinfo{date}{June}
  \bibinfo{year}{2015}).
\newblock
\urldef\tempurl%
\url{https://doi.org/10.2168/LMCS-11(2:7)2015}
\showDOI{\tempurl}


\bibitem[Spinellis(2003)]%
        {CSout}
\bibfield{author}{\bibinfo{person}{Diomidis Spinellis}.}
  \bibinfo{year}{2003}\natexlab{}.
\newblock \bibinfo{title}{{CSout}}.
\newblock
\newblock
\urldef\tempurl%
\url{https://www.spinellis.gr/cscout/}
\showURL{%
\tempurl}


\bibitem[Spinellis(2008)]%
        {spinellis-2008}
\bibfield{author}{\bibinfo{person}{Diomidis Spinellis}.}
  \bibinfo{year}{2008}\natexlab{}.
\newblock \bibinfo{title}{A corrected and annotated version of the
  {X4J11/86-196} document}.
\newblock
\newblock
\urldef\tempurl%
\url{https://www.spinellis.gr/blog/20060626/}
\showURL{%
\tempurl}


\bibitem[Syme(2016)]%
        {fsharp-PR-struct-discriminated-unions}
\bibfield{author}{\bibinfo{person}{Don Syme}.} \bibinfo{year}{2016}\natexlab{}.
\newblock \bibinfo{title}{Fsharp PR 1395: struct discriminated unions}.
\newblock
\newblock
\urldef\tempurl%
\url{https://github.com/dotnet/fsharp/pull/1395}
\showURL{%
\tempurl}


\bibitem[Syme, Neverov and Margetson(2007)]%
        {fsharp-active-patterns}
\bibfield{author}{\bibinfo{person}{Don Syme}, \bibinfo{person}{Gregory
  Neverov}, {and} \bibinfo{person}{James Margetson}.}
  \bibinfo{year}{2007}\natexlab{}.
\newblock \showarticletitle{Extensible Pattern Matching via a Lightweight
  Language Extension}. In \bibinfo{booktitle}{\emph{ICFP'07}}
  \emph{(\bibinfo{series}{ICFP '07})}.
\newblock
\urldef\tempurl%
\url{https://www.microsoft.com/en-us/research/wp-content/uploads/2016/02/p29-syme.pdf}
\showURL{%
\tempurl}


\bibitem[{The C++ standard committee, working group SG12}(2014)]%
        {n3882}
\bibfield{author}{\bibinfo{person}{{The C++ standard committee, working group
  SG12}}.} \bibinfo{year}{2014}\natexlab{}.
\newblock \bibinfo{title}{n3882; An update to the preprocessor specification}.
\newblock
\newblock
\urldef\tempurl%
\url{https://www.open-std.org/jtc1/sc22/wg21/docs/papers/2014/n3882.pdf}
\showURL{%
\tempurl}


\bibitem[{The C standard committee, working group WG14}(1992)]%
        {dr017}
\bibfield{author}{\bibinfo{person}{{The C standard committee, working group
  WG14}}.} \bibinfo{year}{1992}\natexlab{}.
\newblock \bibinfo{title}{Defect report 017}.
\newblock
\newblock
\urldef\tempurl%
\url{https://www.open-std.org/Jtc1/sc22/wg14/www/docs/dr_017.html}
\showURL{%
\tempurl}


\bibitem[Turner(1979)]%
        {turner-79}
\bibfield{author}{\bibinfo{person}{David~A. Turner}.}
  \bibinfo{year}{1979}\natexlab{}.
\newblock \showarticletitle{A new implementation technique for applicative
  languages}. In \bibinfo{booktitle}{\emph{Software - Practice and
  Experience}}.
\newblock


\bibitem[Weeks(2006)]%
        {mlton}
\bibfield{author}{\bibinfo{person}{Stephen Weeks}.}
  \bibinfo{year}{2006}\natexlab{}.
\newblock \showarticletitle{Whole-Program Compilation in MLton}. In
  \bibinfo{booktitle}{\emph{ML Workshop 2006}}.
\newblock
\urldef\tempurl%
\url{http://www.mlton.org/References.attachments/060916-mlton.pdf}
\showURL{%
\tempurl}


\bibitem[Yallop(2020)]%
        {unboxing-RFC}
\bibfield{author}{\bibinfo{person}{Jeremy Yallop}.}
  \bibinfo{year}{2020}\natexlab{}.
\newblock \bibinfo{title}{OCaml RFC: constructor unboxing}.
\newblock
\newblock
\urldef\tempurl%
\url{https://github.com/ocaml/RFCs/pull/14}
\showURL{%
\tempurl}


\end{thebibliography}

\begin{version}{\Extended}
\newpage
\appendix

\section{Correctness and Completeness of our Termination-Checking Algorithm}
\label{appendix:proofs}

\begin{notation}[$\bar t \rewto \bar t'$]
  We often write $\bar t \rewto \bar t'$, reducing terms instead of
  whole programs, when we in fact mean
  $\letrecin D {\bar t} \rewto \letrecin D {\bar t'}$ for some ambient
  environment of recursive definitions $D$.
\end{notation}

We should first establish precisely that our notion of annotated terms
and annotated reduction correspond to an on-the-fly
termination-checking algorithm as we suggested above: we track more
information about the term we are reducing, certain redexes are
blocked, but underneath this is just the standard $\beta$-reduction.

\begin{notation}[$\floor {\bar t}$]
  Let us write $\floor {\bar t}$ (or $\floor {\bar p}$) for the
  unannotated term (or program) obtained from erasing all annotations
  from $\bar t$ (or $\bar p$).
\end{notation}

\begin{fact}
  In a given recursive environment $D$, if a term $\bar t$ that cannot
  take another reduction step, then either:
  \begin{itemize}
  \item It is a $\beta$-normal form in the sense that its erasure
    $\floor {\bar t}$ is a $\beta$-normal form.
  \item Or it contains a \emph{blocked redex} of the form
    $\loc {f (\dots)} l$ with $f \in l$.
  \end{itemize}
\end{fact}

\begin{fact}
The erasure of an annotating substitution is a standard
substitution
\begin{mathline}
  \floor {\locsubs t {\Fam i {x_i \leftarrow \bar t'_i}} l}
  \quad=\quad
  \subs t {\Fam i {x_i \leftarrow \floor {\bar t'_i}}}
\end{mathline}
and thus $\bar t \rewto \bar t'$ implies
$\floor {\bar t} \rewto \floor {\bar t'}$ for the usual notion
$t \rewto t'$ of $\beta$-reduction.
\end{fact}

\begin{corollary}
  If $\bar t$ is reachable and reduces, in the annotated system, to
  a $\beta$-normal form $\bar v$, then $\floor {\bar v}$ is the
  $\beta$-normal form of $\floor {\bar t}$.
\end{corollary}

\subsection{On-The-Fly Termination: Correctness}

We shall now prove \emph{correctness} of our termination-monitoring
approach, that is, that there exist no infinite reduction traces for
the annotated reduction $\bar t \rewto \bar t'$.

We found this to be a non-trivial result, it took us several attempts
to get right. Our most convincing proof is formulated in a more
general setting of well-founded rewriting on $n$-ary trees.

\subsubsection{$n$-ary Trees}

\begin{definition}[Open $n$-ary trees] Given two sets $\param{N}$
  (node constructors) and $\param{V}$ (variables), we define the open
  $n$-ary trees $a \in \TreeSet{\param{N}}{\param{V}}$ by the
  following grammar:
  \begin{mathpar}
    \begin{array}{lcl@{\qquad\qquad}l}
      \TreeSet{\param N}{\param V} \ni a
      & \bnfeq & \node c {\Fam {i \in [0; n]} {a_i}} & (c \in \param{N}) \\
      & \bnfor & \var x & (x \in \param{V}) \\
    \end{array}
  \end{mathpar}
\end{definition}
A tree node is either a node constructor $c \in \param{N}$ followed by
an arbitrary (finite) number of sub-trees, or a variable
$x \in \param{V}$.

\begin{definition}[$\join$]
Open trees have a monadic structure, one can define a $\join$ operation:
\begin{mathpar}
  \join : \TreeSet{\param N}{\param V \uplus \TreeSet{\param N}{\param V}} \to \TreeSet{\param N}{\param V}
\\
  \infer
  {~}
  {\join (\node c {\Fam i {a_i}}) \eqdef \node c {\Fam i {\join(a_i)}}}

  \infer
  {x \in \param V}
  {\join (\var x) \eqdef \var x}

  \infer
  {x \in \TreeSet{\param N}{\param V}}
  {\join (\var x) \eqdef x}
\end{mathpar}
\end{definition}

\begin{definition}[$\pi$, $\atpath \pi b a$, $\SubtreeSet{a}$,
  $\pi . \pi'$, $\pi' \geq \pi$] We define a \emph{path} $\pi$, the
  relation $\atpath \pi b a$ if the path $\pi$ relates the sub-tree
  $b$ to a parent tree $a$, and the multiset $\SubtreeSet a$ of
  sub-trees of $a$.
  \begin{mathpar}
    \begin{array}{llll}
      \pi & \bnfeq & \emptyset & \\
      & \bnfor & \pi . i & \qquad (i \in \mathbb{N}) \\
    \end{array}

    \infer
    {~}
    {\atpath \emptyset a a}

    \infer
    {\atpath \pi b {a_i}}
    {\atpath {\pi . i} b {\node c {\Fam i {a_i}}}}

    \SubtreeSet a \eqdef \dbraces{ b \mid \exists \pi,\ \atpath \pi b a }
  \end{mathpar}

  Notice that the relation $\atpath \pi b a$ is compatible
  with the usual notion of path concatenation $\pi . \pi'$: if
  $\atpath {\pi_{12}} {a_1} {a_2}$ and
  $\atpath {\pi_{23}} {a_2} {a_3}$, then
  $\atpath {\pi_{12} . \pi_{23}} {a_1} {a_3}$.

  Finally, we write that $\pi'$ is an \emph{extension} of $\pi$, using the notation $\pi' \geq \pi$,
  if and only if $\pi'$ is of the form $\pi'' . \pi$.
  \begin{mathline}
    \pi' \geq \pi

    \eqdef

    \exists \pi'',\ \pi' = \pi'' . \pi
  \end{mathline}
\end{definition}

\begin{definition}[$\pathsubs a \pi b$]
  $\pathsubs a \pi b$ is the tree $a$ where the $\pi$-rooted sub-tree is replaced by $b$:
  \begin{mathpar}
    \infer
    {~}
    {\pathsubs a \emptyset b \eqdef b}

    \infer
    {\pathsubs {a_j} \pi b = a'}
    {\pathsubs {\node c {\Fam {i \in I} {a_i}}} {\pi . j} b
     \eqdef
     \node c {\left(\Fam {i \in I \setminus \{j\}} {a_i}, (j \mapsto a')\right)}}
  \end{mathpar}
\end{definition}

We use $n$-ary trees to model our annotated $\lambda$-terms.

\begin{example}[Trees of annotated terms]
  \label{ex:trees-of-annotated-terms}
  Let us consider our annotated terms $\bar t$ as $n$-ary trees
  $\tree(\bar t)$; these trees belong to the set
  $\TreeSet{\param T}{\emptyset}$ for a set of tree node constructors
  $\param T$ defined by the following grammar:
  \begin{mathpar}
    \param T \ni c \bnfeq x \bnfor \loc f l
  \\
    \tree(x)
    \eqdef
    \node {x} \emptyset

    \tree (\loc {f {\Fam i {\bar t_i}}} l)
    \eqdef
    \node {\loc f l} {\Fam i {\tree (\bar t_i)}}
  \end{mathpar}

  For example the annotated term $\loc {f (\loc {g ()} {l_g}, x)} {l_f}$
  becomes the tree
  \begin{mathline}
    \node {\loc f {l_f}} {\left(
        0 \mapsto \node {\loc g {l_g}} \emptyset,
        1 \mapsto \node {x} \emptyset
      \right)}
  \end{mathline}
\end{example}

\subsubsection{Tree Expansion}

We now define an expansion process for our $n$-ary trees, that
generalizes our reduction of terms. Informally, we consider the following expansion process:
  \begin{itemize}
  \item A node (in arbitrary position) in the tree is chosen to be expanded.
    We call \emph{expanded sub-tree} the sub-tree rooted at this node.;
  \item The expansion process replaces this expanded sub-tree by another
    sub-tree containing:
    \begin{itemize}
    \item zero, or several \emph{new} nodes, that do not correspond
      to nodes from the previous tree, and
    \item an arbitrary number of copies of some sub-trees of the expanded sub-tree.
    \end{itemize}
  \end{itemize}
  A sub-tree of the expanded sub-tree can disappear during an
  expansion, or be duplicated into one or several sub-trees, placed
  inside the new sub-tree that replaces the expanded sub-tree.

Formally:
\begin{definition}[Tree expansion]
  A \emph{head tree expansion} $b \hrewto b'$ replaces
  a tree $b$ by another tree $b'$, some sub-trees of which are
  strict sub-trees of $b$ -- not $b$ itself.
  \begin{mathpar}
    \infer
    {b \in \TreeSet{\param N}{\param V}
     \\
     b_s \in \TreeSet{\param N}{\param V \uplus (\SubtreeSet b \setminus \{b\})}}
    {b \hrewto {\join(b_s)}}
  \end{mathpar}

  A \emph{tree expansion} $a \rewto a'$ is a head tree expansion in a sub-tree of $a$.
  \begin{mathpar}
    \infer
    {\atpath {\pi_b} b a
      \\
     b \hrewto b'}
    {a \rewto \pathsubs a {\pi_b} {b'}}
  \end{mathpar}
\end{definition}

\begin{figure}[!htb]
\begin{tikzpicture}[->,>=stealth',level/.style={sibling distance = 1.4cm/#1,
  level distance = 0.7cm}]

\tikzset{
  treenode/.style = {align=center, inner sep=0pt, circle, text centered, font=\sffamily, text width=1em},
  old/.style = {treenode, circle, draw=black},
  new/.style = {treenode, circle, draw=red}
}

\node [old] {$a$}
child{ node [old] {$b$} }
child{ node [old] {$c$}
  child{ node [old] {$d$} }
  child{ node [old] {$e$}
    child{ node [old] {$f$} }
  }
};

\node[xshift=2.5cm,yshift=-1cm] {$\rewto$};

\begin{scope}[xshift=5cm]
\node [old] {$a$}
child{ node [old] {$b$} }
child{ node [new] {$c'$}
  child{ node [old] {$d$} }
  child{ node [old] {$e$}
    child{ node [old] {$f$} }
  }
};
\end{scope}

\begin{scope}[xshift=8cm]
\node [old] {$a$}
child{ node [old] {$b$} }
child{ node [old] {$e$}
  child{ node [old] {$f$} }
};
\end{scope}

\begin{scope}[xshift=11cm]
\node [old] {$a$}
child{ node [old] {$b$} }
child{ node [new] {$c'$}
  child { node [new] {$c''$}
    child{ node [old] {$e$}
      child{ node [old] {$f$} }
    }
    child{ node [old] {$f$} }
  }
  child{ node [old] {$e$}
    child{ node [old] {$f$} }
  }
};
\end{scope}
\end{tikzpicture}

\caption{Tree expansions}
\label{fig:expansions}
\end{figure}
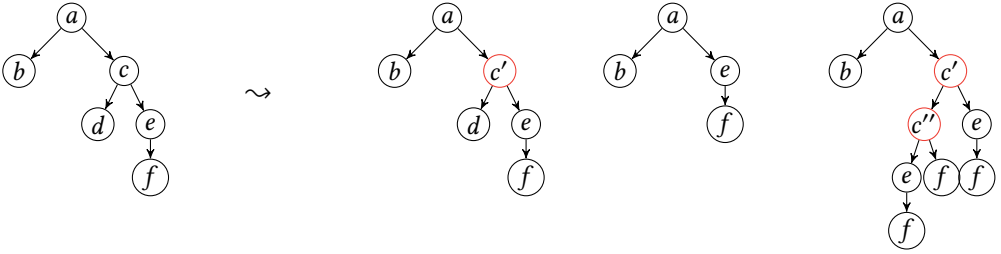

\begin{example}
  Figure~\ref{fig:expansions} shows a tree on the left, and three
  possible expansions of it on the right, all from the node with
  labeled by the constructor $c$, with new nodes shown in red. In the
  first expansion, the node $c$ is replaced by the new node $c'$,
  which carries the children of $c$ (they were neither duplicated
  nor erased). In the second expansion, there is no new node, $c$ has
  been replaced by its sub-tree rooted in $e$ and $d$ was erased. In
  the third expansion, there are two new nodes $c'$ and $c''$, two
  copies of the sub-tree rooted in $e$, one additional copy of the
  sub-tree rooted in $f$, and $d$ is erased: the sub-tree of
  constructor $c$ has been replaced by a sub-tree of the form
  $c'(c''(x,y),x)$, where the variable $x$ is the sub-sub-tree $e(f)$,
  and $y$ is the sub-sub-tree $f$.
\end{example}

\begin{fact}
  The tree expansion process generalizes $\beta$-reduction of our
  annotated first-order terms: for any definition environment $D$,
  \begin{mathpar}
    \letrecin D {\bar t} \rewto \letrecin D {\bar t'}

    \implies

    \tree(\bar t) \rewto \tree(\bar t')
  \end{mathpar}
  In particular, any infinite reduction sequence of annotated terms
  could be mapped to an infinite sequence of tree expansions on the
  corresponding trees.
\end{fact}

\begin{remark}
  $\beta$-reductions in our annotated terms replace variables $x_i$ by
  function parameters $\bar t_i$. In terms of tree, this corresponds
  to a subset of expansions where the erased/duplicated sub-trees are
  always exactly the immediate children of the expanded node, not
  sub-trees of arbitrary depth. We believe that the ability to
  manipulate arbitrary sub-trees, which does not make the proof more
  difficult, is in fact useful to express expansion processes outside
  our model. In particular, it would be useful to model \emph{type
    constraints} in the OCaml programming language:
  \begin{lstlisting}
  type 'a t = ('b * bool) constraint 'a = 'b * 'b * int
  \end{lstlisting}
  which effectively allow to designate arbitrary sub-expressions of
  the type parameters and use them directly in an abbreviation
  definition. In rewriting systems it is also not uncommon to have
  complex left-hand-side that match on their arguments in-depth.
\end{remark}

\begin{definition}[$\CTreeSet{\param N}$] Variables $\param{V}$ were
  used to define expansion by substitution. In the rest of this
  document we will mostly work with \emph{closed} $n$-ary trees in
  $\TreeSet{\param N}{\emptyset}$, which we write as just
  $\CTreeSet{\param N}$.
\end{definition}

\begin{definition}[$\constr{a}$]
  For $a \in \CTreeSet{\param N}$ we define $\constr{a}$ as the head
  constructor of $a$ (which cannot be a variable as $a$ is in
  $\TreeSet{\param N}{\emptyset}$):
\begin{mathline}
  \constr{\node c {\Fam i {a_i}}} \eqdef c
\end{mathline}
\end{definition}

\subsubsection{Measured Expansions}

\begin{definition}[Measure] For a set $A$, we call \emph{measure on A}
  a function $m_A : A \to M$ into a set $M$ equipped with
  a well-founded partial order $(<_M)$.\footnote{A (strict) partial
    order $(<)$ is well-founded if it does not admit any infinite
    descending chain, any sequence $\Fam {i \in \mathbb{N}} {m_i}$
    with $\forall i \in \mathbb{N},\ m_i > m_{i+1}$.}.
\end{definition}

\begin{definition}[Measured expansion]
  Given a measure $m_{\param N}$ on $\param N$, we write that an expansion
  $a \rewto a' \in \CTreeSet{\param N}$ is \emph{measured} when the
  new nodes of the expansion have strictly smaller measure than the
  expanded node. Formally, an expansion
  \begin{mathpar}
    \infer
      {\atpath {\pi_b} b a
        \\
       \infer*
       {b_s \in \TreeSet{\param N}{\SubtreeSet{a} \setminus \{a\}}}
       {b \hrewto {\join(b_s)}}}
      {a \rewto \pathsubs a \pi {\join(b_s)}}
  \end{mathpar}
  is measured if the nodes of $b_s$ are strictly below those of $b$:
  \begin{mathline}
    \forall {\pi},\ \atpath {\pi} {\node c \_} {b_s}
    \implies
    m_{\param N}(c) < m_{\param N}(\constr{b})
  \end{mathline}
\end{definition}

\begin{lemma}
  \label{lem:annotated-terms-measured}
  The reduction of annotated terms is measured.
\end{lemma}

\begin{proof}
  Given a fixed environment $D$ of recursive definitions, we consider
  the subset of tree expansions of the form
  $\tree (\bar t) \rewto \tree (\bar t')$ coming from a reduction of
  annotated terms $\letrecin D {\bar t} \rewto \letrecin D {\bar t'}$.

  Such a reduction of annotated terms is of the form
  \begin{mathpar}
    \infer
    {(f {\Fam i {x_i}} = t') \in D \\ f \notin l}
    {
      \letrecin D C[\loc {f {\Fam i {\bar t_i}}} l]
      \rewto
      \letrecin D C[\locsubs {t'} {\Fam i {x_i \leftarrow \bar t_i}} {l,f}]
    }
  \end{mathpar}

  As a measure of an annotated node $\loc \dots l$, we use the trace
  $l$ ordered by anti-inclusion: $l \leq l'$ if and only if
  $l \supseteq l'$. Our traces may not contain duplicates, and the
  number of function names in $D$ is finite; its cardinality is the
  maximal size of any trace, so the inclusion order is
  well-founded. We extend this order with a minimal element for
  all free variables $x$, that cannot reduce further.

  The replacement of $\loc {f {\Fam i {\bar t_i}}} l$ by
  $\locsubs {t'} {\Fam i {x_i \leftarrow \bar t_i}} {l,f}$ corresponds to
  a measured expansion, as the new nodes are (variables or) annotated
  with the trace $l, f$ which is strictly smaller than $l$ for our
  order.
\end{proof}

\subsection{A Measure on Trees Meant for Measured Expansions}

To prove our termination result, we use the measure $m_{\param N}$ on
nodes to build a measure on trees into a well-chosen set with
a well-founded partial order, such that measured expansions strictly
decrease the tree measure.

Given a set $(S, <_S)$ with a well-founded partial order, we can build
a well-founded partial order $(<_{\Set(S)})$ on the set $\Set(S)$ of
sets of elements of $S$, and a well-founded partial order
$(<_{\MSet(S)})$ on the set $\MSet(S)$ of multisets of elements of
$S$. The multiset order is attributed to
\href{https://en.wikipedia.org/wiki/Dershowitz-Manna_ordering}{Dershowitz,
  Manna, Huet and Oppen}).
\begin{itemize}
\item $S_1 <_{\Set(S)} S_2$ if, for any $a \in S_1$, there exists $b \in S_2$ such that $a <_S b$.
\item $M_1 <_{\MSet(S)} M_2$ if $M_1$ is of the form $M + N_1$
  and $M_2$ is of the form $M + N_2$ with
  $\msettoset{N_1} <_{\Set(S)} \msettoset{N_2}$.
\end{itemize}

Given a measure $m_{\param N} : \param N \to (O, <_O)$ on the node
constructors of our tree (where $(<_O)$ is a well-founded
partial order), we define for any closed tree
$a \in \CTreeSet{\param N}$ a measure
$m_{\Node a} : \SubtreeSet{a} \to (\MSet(O), <_{\MSet(O)})$ on the
nodes of $a$, by considering for each node the path from this node
(included) to the root of the tree, seen as a multiset of nodes.

Formally, we define our measure $m_{\Node a}(b)$ into $\MSet(O)$, for
each node $b \in \SubtreeSet{a}$, as the multiset sum of the measure
of the constructor of $b$ and a measure $m_{\Path{a}}$ also into
$\MSet(O)$ defined by induction on the predicate $\atpath \pi b a$.
\begin{mathpar}
  m_{\Node a} (b)
  \quad\eqdef\quad
  \dbraces {m_{\param N} (\constr b)} + m_{\Path{a}}(\atpath \pi b a)

  m_{\Path{a}} \left(
    \infer
    {}
    {\atpath \emptyset a a}
  \right) \eqdef \emptyset

  m_{\Path{\node c {\Fam i {a_i}}}} \left(
  \infer
  {\atpath \pi b {a_i}}
  {\atpath {\pi . i} b {\node c {\Fam i {a_i}}}}
  \right)
  \eqdef
    \dbraces{m_{\param N}(c)}
    +
    m_{\Path{a_i}}\left(\atpath \pi b {a_i} \right)
\end{mathpar}

Note that the measure of the root $a$ is not taken into account by our
path measure: we define $m_{\Path{a}} (\atpath \emptyset a a)$ as
$\emptyset$ rather than $\dbraces{m_{\param N}(a)}$. This is essential
for the measure to be compatible with path concatenation:

\begin{fact}
\begin{mathpar}
  \atpath {\pi_{12}} {a_1} {a_2}
  \quad\wedge\quad
  \atpath {\pi_{23}} {a_2} {a_3}
  \\
  \implies

  m_{\Path {a_3}}(\atpath {\pi_{12}.\pi_{23}} {a_1} {a_3})
  \quad=\quad
  m_{\Path {a_2}}(\atpath {\pi_{12}} {a_1} {a_2})
  +
  m_{\Path {a_3}}(\atpath {\pi_{23}} {a_2} {a_3})
\end{mathpar}
\end{fact}

From this measure on nodes
$m_{\Node{a}} : \SubtreeSet{a} \to (\MSet(O), <_{\MSet(O)})$
we now define a measure on trees
$m_\Tree : \CTreeSet{\param N} \to (\MSet(\MSet(O)), <_{\MSet(\MSet(O))})$
by considering for each tree the multiset of measures of its nodes.
Our trees are thus measured as multisets of multisets of node
constructors, themselves equipped with a well-founded partial order.
\begin{mathline}
  m_\Tree(a) \eqdef \dbraces {m_{\Node a}(b) \mid b \in \SubtreeSet a }
\end{mathline}

\subsection{Termination of Measured Expansions}

\begin{theorem}[Termination]\label{thm:termination}
  If $\param M$ has a measure $m_{\param N}$, then the tree expansions
  that are measured for $m_{\param N}$ are strictly decreasing for the
  order $m_\Tree$. In particular, there are no infinite sequences of
  measured expansions.
\end{theorem}

\begin{proof}
A measured expansion $a \rewto a'$ is necessarily of the following
form, where a sub-tree $b$ at a path $\pi_b$ in $a$ is replaced by
a sub-tree $\join(b_s)$, where the constructors of $b_s$ are of
strictly smaller measure than $b$:

\begin{mathpar}
  \infer
  {b_s \in \TreeSet{\param N}{\SubtreeSet{a} \setminus \{a\}}}
  {\infer
    {\atpath {\pi_b} b a
      \\ b \hrewto {\join(b_s)}}
    {a \rewto \pathsubs a {\pi_b} {\join(b_s)} = a'}}

  \forall \pi, \atpath \pi {(\node c \_)} {b_s}
  \implies
  m_{\param N}(c) < m_{\param N}(\constr{b})
\end{mathpar}

We have to show that $m_\Tree(a) > m_\Tree(a')$.

\paragraph{Nodes inside or outside $\pi_b$}

We can partition the paths $\pi$ valid inside $a$ or $a'$ in two
disjoint sets:
\begin{itemize}
\item the paths to nodes \emph{inside} $\pi_b$, that is the paths
  $\pi \geq \pi_b$ that extend $\pi_b$; they denote a sub-tree of $b$
  (in $a$) or a sub-tree of $\join(b_s)$ (in $a'$), and
\item the paths to nodes \emph{outside} $\pi_b$, that is the paths
  $\pi \ngeq \pi_b$ that do not extend $\pi_b$, they reach an
  unrelated part of the tree.
\end{itemize}
\begin{mathpar}
  \begin{array}{ll}
    & m_\Tree(a) \\
    = & \dbraces {m_{\Node a}(b') \mid \exists b' \in \SubtreeSet a } \\
    = & \dbraces {m_{\Node a}(b') \mid \exists b' \exists \pi,\ \atpath \pi {b'} a } \\
    = &
        \dbraces {m_{\Node a}(b') \mid \exists b' \exists \pi \geq \pi_b,\ \atpath \pi {b'} a }
        +
        \dbraces {m_{\Node a}({b'}) \mid \exists b' \exists \pi \ngeq \pi_b,\ \atpath \pi {b'} a }
    \\[1em]

    & m_\Tree(a') \\
    = &
        \dbraces {m_{\Node {a'}}({b'}) \mid \exists b' \exists \pi \geq \pi_b,\ \atpath \pi {b'} {a'} }
        +
        \dbraces {m_{\Node {a'}}({b'}) \mid \exists b' \exists \pi \ngeq \pi_b,\ \atpath \pi {b'} {a'} }
    \\
  \end{array}
\end{mathpar}

The measure $m_{\Node a}(b')$ of a sub-tree $b'$ of $a$ depends only
on the constructors along the path from the root $a$ to the node
$b'$ included. In particular, if the node $b'$ is at a path
$\pi \ngeq \pi_b$ that does not extend $\pi_b$, it is also present
at the same path as a sub-tree of $a'$, where it has the same
measure, as $a'$ and $a$ only differ inside $\pi_b$.
\begin{mathline}
  \dbraces {m_{\Node a}({b'}) \mid \exists b' \exists \pi \ngeq \pi_b,\ \atpath \pi {b'} a }

  =

  \dbraces {m_{\Node {a'}}({b'}) \mid \exists b' \exists \pi \ngeq \pi_b,\ \atpath \pi {b'} {a'} }
\end{mathline}

Recall that our multiset order is defined by inequalities
$M + N_1 <_{\MSet(S)} M + N_2$ with
$\msettoset{N_1} <_{\Set(S)} \msettoset{N_2}$. We use the measures
of the nodes outside $\pi_b$ as the multiset $M$ in that
definition. It remains to prove the set ordering between the
measures of nodes inside $\pi_b$, playing the role of $N_1$ and
$N_2$ in the definition:
\begin{mathline}
  \braces {m_{\Node a}(b') \mid \exists b' \exists \pi \geq \pi_b,\ \atpath \pi {b'} a }

  >_{\Set(\MSet(O))}

  \braces {m_{\Node {a'}}({b'}) \mid \exists b' \exists \pi \geq \pi_b,\ \atpath \pi {b'} {a'} }
\end{mathline}

\paragraph{A classification of paths under $\join(b_s)$}

For each sub-tree under $\pi_b$ in $a'$, that is each sub-tree of
$\join(b_s)$, we have to show that there exists a sub-tree under
$\pi_b$ in $a$, that is a sub-tree of $b$, of strictly larger
measure.

Let us remark (and prove) that a sub-tree $b'$ of $a'$ has a path of
the form $\pi_s . \pi_n . \pi_b$, with, reading right to left from
the root to the leaves:
\begin{itemize}
\item the path $\pi_b$ from $b$ to the root $a$,
\item a path $\pi_n$ (\emph{n}ew node) inside zero, one or several
  ``new'' nodes introduced by the expansion,
\item a path $\pi_s$ (\emph{s}ubtree), possibly empty,
  inside a copy of a strict subtree of $b$.
\end{itemize}
When $\pi_s$ does indeed reach inside a subtree copied from $b$,
we furthermore define a path $\pi_o$ (\emph{o}ld node) from $b$ to the first
old node copied -- otherwise $\pi_v$ is defined to be empty.

For example, in Figure~\ref{fig:expansions}, some paths into the
nodes of the third expansion of $c$ are as follows -- for readability we
write paths by marking the constructor of the node, rather than its
child index:
\begin{itemize}
\item $[a, c']$, which has the decomposition $\pi_b = [a]$, $\pi_n = [c']$,
  $\pi_s = []$ with $\pi_o = []$.
\item $[a, c', c'', f]$, which has the decomposition $\pi_b = [a]$,
  $\pi_n = [c', c'']$, $\pi_s = [f]$ with $\pi_v = [e]$.
\end{itemize}

The presence of $\pi_b$ is by hypothesis, we are classifying the
paths $\pi \geq \pi_b$ which necessarily have $\pi = \pi' . \pi_b$
for some $\pi'$ such that
$\atpath {\pi'} {b'} {\join(b_s)}$. Formally, we just defined how to
split this sub-path $\pi'$ into a pair $(\pi_s, \pi_n)$ by defining
a function $\splitjoin {b_s} {\atpath {\pi'} {b'} {\join(b_s)}}$ by induction
on $b_s$ and inversion on the definition of $\join(b_s)$:
\begin{mathpar}
  \splitjoin{\var x}{
    \infer
    {x \in \TreeSet{\param N}{\param V}}
    {\atpath {\pi''} {b'} \join (\var x) = x}
  } \eqdef (\pi'', \emptyset)

  \splitjoin{\node c \_}{
    \infer
    {~}
    {\atpath \emptyset {b'} {\node c \_ = b'}}
  } \eqdef (\emptyset, \emptyset)

  \begin{array}{r}
  \splitjoin{\node c {\Fam i {a_i}}}{
    \infer
    {\atpath {\pi''} {b'} {\join (a_i)}}
    {\atpath {\pi'' . i} {b'} {\join (\node c {\Fam i {a_i}})}}
  } \eqdef (\pi''_s, \pi''_n . i)
  \\ \quad\text{where}\quad (\pi''_s, \pi''_n) = \splitjoin {a_i} {\atpath {\pi''} {b'} {\join (a_i)}}
  \end{array}
\end{mathpar}

By definition we have $\pi_s.\pi_n = \pi'$, so $\atpath {\pi_s . \pi_n} {b'} {\join(b_s)}$, and we also have
 $\atpath {\pi_s} {b'} {\join(b_s')}$ and $\atpath {\pi_n} {b_s'} {b_s}$ for a sub-tree
$b_s'$ of $b_s$ such that:
\begin{enumerate}
\item[(1)] Either $b'$ is not of the form $\var x$: the path $\pi'$
  stops inside $b_s$ before reaching a node $\var x$. In this case
  $b'$ is a ``new node'' of $b_s$ -- in particular, we know that
  expansion generated at least one new node, otherwise we would have
  $b_s = \var \_$. On have $\pi_s = \emptyset$ by definition of
  $\splitjoin \_ \_$, and we additionally define $\pi_o$ as
  $\emptyset$: there is no path to a corresponding ``old node'' in
  $b$.
\item[(2)] Or the path $\pi'$ reaches a node $b_s' = \var {b''}$ of
  $b_s$: $\pi_n$ is the path between $b_s'$ and $b_s$, and the rest of
  $\pi'$ is a path $\atpath {\pi_s} {b'} {b''}$, possibly empty,
  reaching a strict sub-tree $b''$ of $b$. We then define $\pi_o$ as
  the path of $b''$ in $b$, non-empty as $b''$ is a strict sub-tree of
  $b$.  strict.
\end{enumerate}
We will come back to this distinction between those two cases (1) and
(2) in the rest of the proof.

We shall now conclude the proof by finding a sub-tree of $a$ under
$\pi_b$ (thus inside $b$) whose measure is strictly larger than the
measure of our sub-tree $b'$ of path $\pi_s . \pi_n . \pi_b$ inside
$a'$.

\paragraph{Larger path}

Our candidate sub-tree, that we name $b'_v$, is at the path
$\pi_s . \pi_o . \pi_b$ in $a$.

\paragraph{Case (1)} In the case (1) above
($\pi_s = \pi_o = \emptyset$), $b'$ is a new node, and $b'_o$
is the sub-tree at path $\emptyset . \emptyset . \pi_b$, that is
$\pi_b$: the node $b'_o$ is exactly $b$. We have to show
$m_{\Node{a}}(b) > m_{\Node{a'}}(b')$.
\begin{mathpar}
  \begin{array}{l}
    m_{\Node{a}}(b'_o)
    \\ =
    m_{\Node{a}}(b)
    \\ =
    m_{\Path{a}}(\pi_b) + \dbraces{ m_{\param N}(\constr{b}) }
    \\[1em]
  \end{array}

  ~

  \begin{array}{l}
    m_{\Node{a'}}(b')
    \\ =
    m_{\Path{a'}}(\pi_n . \pi_b) + \dbraces{ m_{\param N}(\constr{b'}) }
    \\ =
    m_{\Path{\join(b_s)}}(\pi_n) + m_{\Path{a'}}(\pi_b) + \dbraces{ m_{\param N}(\constr{b'}) }
    \\ =
    m_{\Path{\join(b_s)}}(\pi_n) + m_{\Path{a}}(\pi_b) + \dbraces{ m_{\param N}(\constr{b'}) }
  \end{array}
\end{mathpar}
The expansion from $a$ to $a'$ is measured, so we have
$m_{\param N}(\constr{b}) > m_{\param N}(\constr{b'})$, which concludes the proof in this case.

\paragraph{Case (2)} In the case (2) above, $b'$ is a sub-tree of
$b''$, a strict sub-tree of $b$ copied by the expansion generated by $b_s$.
$\pi_s$ is the path of $b'$ inside $b''$, and $\pi_o$
is the (non-empty) path of $b''$ inside $b$. We have:
\begin{mathpar}
  \begin{array}{ccccccccc}
  m_{\Node{a}}(b'_o)
  & = & m_{\Path{a}}(\pi_b)
  & + & m_{\Path{b}}(\pi_o)
  & + & m_{\Path{b''}}(\pi_s)
  & + & m_{\param N}(b'_o)
  \\[.6em]

  & 
  & (=) 
  & 
  &
  & 
  & (=) 
  & 
  & (=) 
  \\[.6em]

  m_{\Node{a'}}(b')
  & = & m_{\Path{a'}}(\pi_b)
  & + & m_{\Path{b_s}}(\pi_n)
  & + & m_{\Path{b''}}(\pi_s)
  & + & m_{\param N}(b')
  \\
  \end{array}
\end{mathpar}
To conclude we have to show
$m_{\Path{b}}(\pi_o) > m_{\Path{b_s}}(\pi_n)$. This is true because
$\pi_o$ is non-empty, so $m_{\Path{b}}(\pi_o)$ is a sum that contains in particular
the measure of its root $m_{\param N}(b)$, which is strictly larger
than the measure of all the nodes of $b_s$ as the expansion is measured.
\end{proof}

\begin{corollary}
  Annotated reduction is strongly normalizing.
\end{corollary}

We have proved that our on-the-fly termination checking strategy is
\emph{sound}, in the sense that it does enforce termination: each
reduction sequence terminates in a finite number of step, either
because it reaches a normal form (also normal as an
un-annotated term), or because it is blocked with all reducible
position trying to reduce a function already in the trace.

\subsection{On-The-Fly Termination: Completeness}

In this section we establish \emph{completeness} of our on-the-fly
termination checking strategy, in the sense that it does not prevent
any normalizing term from reaching its normal form.

\begin{theorem}[Weak completeness]
  If a reachable annotated program $\letrecin D {\bar t}$ contains
  a blocked redex $\loc {f (\Fam i {\bar t'_i})} l$ with $f \in l$,
  then its unannotated erasure $\floor {\bar t}$ admits an infinite
  reduction sequence.
\end{theorem}

\begin{proof}
  To make notations lighter, we will assume that $f$ takes a single
  argument. This does not change the proof argument in the least.

  Our proof proceed in three steps.
  \begin{enumerate}
  \item We work backwards from the blocked redex $f (t')$ with a trace
    of the form $l_1, f, l_2$, we rewind the trace, to exhibit
    a previous redex $f(u)$ with trace $l_1$ that reduces in at least
    one step to our term $C'[f(t')]$:
    \begin{mathline}
      f(u) \rewto \locsubs {t_{f}} {x \leftarrow u} {l,f} \rewto^\star C'[f(t')]
    \end{mathline}
    At this point we can erase the trace annotations and consider the
    corresponding unannotated reduction sequence.
  \item We \emph{generalize} this reduction sequence to a valid,
    non-empty reduction sequence that does not depend on $u$, with
    instead a fresh variable $x$:
    \begin{mathline}
      f(x) \rewto t_{f} \rewto^+ C''[f(t'')]
    \end{mathline}
  \item At this point we can build an infinite reduction on unannotated terms
    by repeating this generalized reduction sequence inside its own conclusion:
    \begin{mathline}
      f \floor{t'}
      \rewto^+
      \subs {C''\left[f(t'')\right]} {x \leftarrow t'}
      \rewto^+
      \subs {C''\left[\subs {C''\left[f(t'')\right]} {x \leftarrow t''}\right]} {x \leftarrow t'}
      \rewto^\infty
      \dots
    \end{mathline}
  \end{enumerate}

  \paragraph{Step 1: rewinding the trace}

  We have a blocked redex of the form $\loc {f (t')} {(l_1, f, l_2)}$
  that is reachable, that is (Definition~\ref{def:reachable}), that
  appears a reduction sequence starting from a term
  $\locsubs t \emptyset \emptyset$ with empty traces -- we consider
  the recursive environment $D$ fixed once and for all in this proof.

  We claim that if a reachable term $\bar t$ has a trace of the form
  $l, f'$ for some function call $f'$ (let us assume $f'$ unary again
  for simplicity), then it necessarily arises from the reduction of
  a call of the function $f'$, followed possibly by some further
  reduction steps. More precisely, there exists a non-empty reduction
  sequence:
  \begin{mathline}
    \loc {f'(u)} l \rewto \locsubs {t_{f'}} {x \leftarrow u} {l,f'} \rewto^\star C[\bar t]
  \end{mathline}
  We show this by induction on the reduction sequence reaching
  $\bar t$ (which must be non-empty as it has a non-empty trace),
  doing a case analysis on the reductions that can result in a term of
  the form $C[\loc {\bar t} {(l, f')}]$:
  \begin{itemize}
  \item If the last reduction was precisely the call to $f'$ whose
    expansion generated $\bar t$, we are done.
  \item If the last reduction occurred in a strict subterm of
    $\bar t$, then its source term $\bar t'$ also has trace $l, f'$;
    we use our induction hypothesis to get a reduction from
    a $\loc {f'(u)} l$ to $C[\bar t']$ and complete with the last
    reduction step.
  \item If the last reduction step occurred in the context outside
    $\bar t$ without touching $\bar t$, we forget it and apply the
    induction hypothesis on the strictly smaller (but still non-empty)
    reachable sequence.
  \item Finally, if the last reduction step substituted $\bar t$
    (unchanged) from a different position in the term, our induction
    hypothesis immediately gives us a suitable reduction sequence
    reaching this earlier occurrence of $\bar t$, and we are done.
  \end{itemize}

  Now that we know how to rewind the last element from the trace, we
  can iterate this process on the trace $l_1, f, l_2$ of our blocked
  redex $f(t')$, once per element of $l_2$ and once for $f$. We thus
  get a non-empty, $f$-repeating reduction sequence
  \begin{mathline}
    \loc {f(u)} {l_1} \rewto \locsubs {t_f} {x \leftarrow u} {l_1, f} \rewto^\star C'[\loc {f(t')} {(l_1, f, l_2)}]
  \end{mathline}

  \paragraph{Let us drop the annotations now}

  We crucially used our trace annotations to build this $f$-repeating
  reduction sequence. In the rest of the proof we will not use the
  annotations anymore, so we move to the lighter world on
  unannotated terms: we redefine
  \begin{mathline}
    u \eqdef \floor u

    C' \eqdef \floor {C'}

    t' \eqdef \floor {t'}
  \end{mathline}
  and consider an unannotated, non-empty, $f$-repeating reduction sequence:
  \begin{mathline}
    f(u) \rewto \subs {t_f} {x \leftarrow u} \rewto^\star C'[f(t')]
  \end{mathline}

  \paragraph{Step 2: generalizing the $f$-repeating sequence}
  We claim that our $f$-repeating reduction sequence (now unannotated)
  starting in $f(u)$ can be generalized in a sequence starting in
  $f(x)$ for any fresh variable $x$.

  This proof step relies on the first-order restriction on our
  $\lambda$-calculus: this works easily because substitution $u$
  inside the body of $f$ cannot create new redexes. After the
  substitution, we can reduce redexes in copies of $u$ (if any), or
  reduce redexes already present in the body of $f$ (possibly with
  a variable $x$ now replaced by $u$), but substituting $u$ does not
  make new function calls possible.

  In consequence, we can systematically replace $u$ by a fresh
  variable $x$ in our reduction sequence, dropping any reduction step
  that is internal to $u$. Notice that the first step in our sequence
  \begin{mathline}
    f(u) \rewto \subs {t_f} {x \leftarrow u} \rewto^\star C'[f(t')]
  \end{mathline}
  is \emph{not} internal to $u$, so we know that it remains in the
  filtered reduction sequence -- we need it non-empty. Modulo
  $\alpha$-renaming we will assume that our fresh variable $x$ is also
  the name of the bound variable in the declaration of $f$. We get
  a filtered sequence of the form:
  \begin{mathline}
    f(x) \rewto t_f \rewto^\star C''[f(t'')]
  \end{mathline}

  \paragraph{Step 3: profit}

  Our generalized $f$-repeating reduction sequence starts from $f(x)$
  for a fresh variable $x$ and, in at least one step, reduces to
  a term that contains a call to $f$ again. We can immediately build
  an infinite reduction sequence by repeatedly composing our
  generalized sequence with itself.
\end{proof}

What we have shown so far is that blocked redexes could generate
infinite reduction sequences. But our reduction is non-deterministic,
so some terms have both terminating and non-terminating reduction
sequences -- they are normalizing but not strongly normalizing. Using
standard(ization) results on the $\lambda$-calculus we believe that we
can get a stronger completeness result -- the proof below is a bit
sketchy, so we degraded our claim from Theorem to Conjecture.

\begin{conjecture}[Strong completeness]
  If an annotated program $\letrecin D {\bar t}$ has a weakly
  normalizing erasure, that is, if $\floor t$ has \emph{some}
  reduction path to a normal form, then $\bar t$ has an annotated
  reduction sequence to a (non-blocked) $\beta$-normal form.
\end{conjecture}

\begin{proof}
  A standard fact about the $\lambda$-calculus, which applies to our
  setting, is that the leftmost-outermost reduction strategy is
  minimally diverging, in particular it reaches the normal forms of
  any weakly normalizing term.

  Our argument relies on studying the annotated leftmost-outermost
  reduction sequence of the annotated term $\bar t$ (under the
  definitions $D$). If this reduction sequence does not reach
  a blocked redex in leftmost-outermost position, we are done: it
  reaches a $\beta$-normal form of $\floor {\bar t}$.

  We now argue that $\bar t$ \emph{cannot} reach a blocked redex by
  following the leftmost-outermost strategy. If $\bar t$ could reach
  a blocked redex, then the proof of the previous theorem would
  provide an infinite reduction sequence from that point. This
  infinite reduction sequence is obtained by repeating a suffix of the
  reduction sequence to this block redex, which follows the
  leftmost-outermost strategy. In consequence, the infinite reduction
  sequence is itself a leftmost-outermost sequence. This is in
  contradiction with the assumption that $\floor t$ is weakly
  normalizing.
\end{proof}

Finally, it is interesting to study how those proofs break down when
we move from our first-order calculus its closed-higher-order
extension presented in
Section~\ref{subsec:closed-higher-order-calculus}. In that section, we
show that completeness does not hold in the closed-higher-order
fragment, with the following counterexample:
\begin{lstlisting}
let rec f(p,q) = p(f(q,q))
    and id(x) = x
    and stop(x) = done
in f(id, stop) (* $\rewto$ id(f(stop,stop)) $\rewto$ f(stop,stop)
                    $\rewto$ stop(f(stop,stop)) $\rewto$ done *)
\end{lstlisting}

The first blocked redex in this reduction sequence is the application
of \code{f} in the second term \code{id(f(stop,stop))} of the
sequence. Using the reasoning of our proof of weak completeness, we can
``generalize'' the $f$-repeating sequence to a sequence of the form
\code{f(k,stop) $\rewto$ k(f(stop,stop))} for any \code{k}, whose
instantiations can be repeated into an infinite reduction sequence:
\begin{lstlisting}
f(id,stop) $\rewto$ id(f(stop,stop)) $\rewto$ id(stop(f(stop,stop))) $\rewto$ id(stop(stop(...)))
\end{lstlisting}
On the other hand, this first occurrence of a blocked redex is
\emph{not} in head position for the leftmost-outermost reduction
strategy, so our proof of the strong completeness theorem does not
applies.

Another blocked redex does appear in head position in the third term
of the reduction sequence, \code{f(stop,stop)}. But then the reasoning
of our proof of weak completeness cannot be applied anymore: the reduction
sequence
\begin{lstlisting}
f(id,stop) $\rewto$ id(f(stop,stop)) $\rewto$ f(stop,stop)
\end{lstlisting}
crucially uses a parameter of \code{f} to create a new redex, and this
sequence cannot be generalized by replacing \code{id} by an arbitrary
parameter {k}.

\subsection{Proof of Correctness for the Closed-Higher-Order Fragment}
\label{app:closed-higher-order-correctness-proof}

This is a proof missing from \ref{subsec:closed-higher-order-calculus}.

\begin{theorem*}
  Our termination-monitoring algorithm remains correct for this
  closed-higher-order language: the annotated language is strongly
  normalizing.
\end{theorem*}

\begin{proof}
  The proof uses the same machinery as the first-order proof -- see
  \appendixref{appendix:proofs}. We translate our terms to trees with
  measured expansions. We can then immediately reuse our result that
  measured expansions are strongly normalizing.

  Let us recall the translation of our higher-order terms into trees
  (Example~\ref{ex:trees-of-annotated-terms}):
  \begin{mathpar}
    \param T \ni c \bnfeq x \bnfor \loc f l
  \\
    \tree(x)
    \eqdef
    \node {x} \emptyset

    \tree (\loc {f {\Fam i {\bar t_i}}} l)
    \eqdef
    \node {\loc f l} {\Fam i {\tree (\bar t_i)}}
  \end{mathpar}
  This translation would use the function name $f$ directly as the
  node constructor for an application node.

  We now use an explicit $\mathsf{app}$ constructor for application
  nodes, as the function name is not known at this point. The function
  term $t$ becomes the first child of the application node;
  \begin{mathpar}
    \param T \ni c \bnfeq x \bnfor f \bnfor \loc {\mathsf{app}} l
    \\
    \tree(x)
    \eqdef
    \node {x} \emptyset

    \tree(f)
    \eqdef
    \node {f} \emptyset

    \tree (\loc {t {\Fam i {\bar u_i}}} l)
    \eqdef
    \node {\loc {\mathsf{app}} l} {t :: \Fam i {\tree (\bar u_i)}}
  \end{mathpar}

  To conclude the proof we only have to remark that reductions of our
  closed-higher-order terms become measured expansions of the
  corresponding trees, extending
  Lemma~\ref{lem:annotated-terms-measured} to the closed-higher-order
  case. For each $f(\Fam i {x_i}) = t'$ in the ambient set of
  recursive definitions $D$, we have tree expansions of the form
  \begin{mathpar}
    \node {\loc {\mathsf{app}} l} {\node {f} \emptyset :: \Fam i {a_i}}
    \rewto
    \subs {\tree(\locsubs {t'} {\emptyset} {l,f})} {\Fam i {x_i \leftarrow a_i}}
  \end{mathpar}
  If $t'$ is a variable, then this reduces to a strict sub-tree, one of
  the arguments $a_i$. If $t'$ is a function name $f$, this reduces to
  a leaf node. Otherwise $t'$ starts with an application node, which
  gets annotated with the strictly larger trace $l,f$ in the resulting
  tree. In all three cases, this is a measured tree expansion.
\end{proof}

\section{More on \code{cpp}}
\label{appendix:cpp}

\subsection{Relating our Algorithm to Dave Prosser's Pseudo-Code}
\label{appendix:relating-to-dave-prosser}

The pseudo-code defines two mutually-recursive functions,
\code{expand} which computes the macro expansion of a full token list,
and \code{subst} which expands a single macro call. Let us quote the
pseudo-code, slightly simplified, with all parts related to irrelevant
\code{cpp} features (stringization, concatenation) elided:

\begin{lstlisting}
expand(TS)
{
  if TS is {} then
    return {};
  else if TS is T^HS $\cdot$ TS' and T is in HS then
    return T^HS $\cdot$ expand(TS');
  else if TS is T^HS $\cdot$ ( $\cdot$ TS' and T is a function-like macro then
    check TS' is Actuals $\cdot$ )^HS' $\cdot$ TS'' and actuals are correct for T,
    return expand(subst(ts(T), fp(T), Actuals, (HS $\cap$ HS') $\uplus$ {T}, {}) $\cdot$ TS'');
  else TS must be T^HS $\cdot$ TS'
    return T^HS $\cdot$ expand(TS');
}
\end{lstlisting}

\code{expand(TS)} processes a token sequence \code{TS} by expanding
macro invocations. (\code{TS $\cdot$ TS'} is concatenation of token
sequences.) The tokens \code{T^HS} of the sequence are raw tokens
\code{T} annotated with a ``hide set'' \code{HS}, which is a set of
macro names. Hide sets \code{HS} correspond to our traces $l$, and
token sequences \code{TS} to our annotated terms $\bar t$.

In the definition of \code{expand}, an annotated macro name
\code{T^HS} is \emph{not} expanded if its name occurs in its own
trace. This corresponds to our blocked redexes. If $T$ is not blocked
by the hide set, then the function \code{subst} is called to perform
a $\beta$-reduction that also expands the macro parameters, and
expansion is rerun on the result. \code{ts(T)} is the code/body of
the macro \code{T}, \code{fp(T)} is an array of formal parameters,
\code{Actuals} an array of actual parameters, the hide set of the
invocation adds \code{T} as expected. It computes an intersection of
\code{HS} and \code{HS'} instead of just using \code{HS}; this only
matters in the closed-higher-order case, which we discuss in another
section.

\begin{lstlisting}
subst(IS, FP, AP, HS, OS)
{
  if IS is {} then
    return hsadd(HS, OS);
  else if IS is T $\cdot$ IS' and T is FP[i] then
    return subst(IS', FP, AP, HS, OS $\cdot$ expand(AP[i]));
  else IS must be T^HS' $\cdot$ IS'
    return subst(IS', FP, AP, HS, OS $\cdot$ T^HS');
}
\end{lstlisting}

The function \code{subst(IS, FP, AP, HS, OS)} computes a single macro
invocation by processing the token sequence \code{IS} coming from the
function body, replacing the formal parameters from \code{FP} by the
expansion of the actual parameters \code{AP}, with hide set \code{HS};
it appends its results to the accumulator token sequence \code{OS} --
which is annotated with hide sets.

In the definition of \code{subst}, formal parameters found in the
macro body are replaced by the recursive expansion of the actual
parameters. Note that the name \code{T} of the macro being expanded
was not added to the hide set of the actual parameters. After the
whole macro body \code{IS} has been traversed, we call the auxiliary
function \code{hsadd} on the output accumulator \code{OS}; this
function arguments the hide set of all tokens of \code{OS} by the
invocation-site hide set \code{HS}.

(Note: we are not sure why \code{IS} is taken to be annotated with
hide sets, while it comes from a macro definition that should not be
the result of any expansion yet. As far as we can tell \code{HS'}
should always be empty in this function, and this is the behavior of
Spinellis' CSout~\citep*{CSout} implementation.)

This \code{subst} function is related to our annotated substitution
$\locsubs t \sigma l$, but not exactly the same: our definition only
uses the trace $l$ for the parts of the term coming from the function
body $t$, not for the substituted parameters in $\sigma$ -- which
remain unreduced. In contrast \code{subst} adds the hide set \code{HS}
to all the output, but only \emph{after} recursively expanding the
actual parameters. In our algorithm this would correspond to first
reducing the arguments in $\sigma$ to a normal form, and then
extending their trace with $l$.

\subsection{Relating the Algorithms on the First-Order Fragment}

For a precise comparison, here is our own rule for function calls, and
a big-step presentation of Dave Prosser's algorithm:
\begin{mathpar}
\infer[our-rule]
{(f {\Fam i {x_i}} = t') \in D \\ f \notin l}
{
  \letrecin D C[\loc {f {\Fam i {\bar t_i}}} l]
  \rewto
  \letrecin D C[\locsubs {t'} {\Fam i {x_i \leftarrow \bar t_i}} {l,f}]
}

\infer[Prosser-rule]
{(f {\Fam i {x_i}} = t') \in D
  \\ f \notin l
  \\ \forall i, \bar t_i \Downarrow_D \bar u_i
  \\ t'' \eqdef (\locsubs {t'} {\Fam i {x_i \leftarrow \bar u_i}} \emptyset) ^{l' \mapsto l' \cup l,f}
  \\ C[t''] \Downarrow_D u'
}
{
  C[\loc {f {\Fam i {\bar t_i}}} l]
  \Downarrow_D
  u'
}
\end{mathpar}

In the definition of $t''$ we use a notation
${\bar t}^{l' \mapsto \dots}$, which we mean as a map over the
annotations of $\bar t$, transforming the trace $l'$ of each syntactic
node into another trace.

We claim that the two approaches give the same results on our
first-order calculus, or equivalently on the first-order macro
fragment: the erasure of the normal forms on both side are the
same. Let us review the two differences between the two rules:
\begin{itemize}
\item The normalization of the $\bar t_i$ into $\bar u_i$ at
  substitution time in Prosser's rule corresponds to a call-by-value
  reduction strategy. Our small-step reduction can model any reduction
  strategy, but it eventually computes the same normal forms.
\item As we discussed, the evolution of traces differ: Prosser's rule
  adds the current trace $l,f$ to all the nodes of the substitution
  $\subs {t'} {\Fam i {x_i \leftarrow \bar u_i}}$, while we only use
  it for the unannotated syntactic nodes coming from $t'$. The
  annotations on the nodes of $u_i$ thus differ in the final result --
  this is the only difference. But this does not change the reduction
  dynamics, given that the $u_i$ are already in normal form: no
  further reduction happens in them, so in particular there are not
  more blocked redexes. After erasure we get the same results.
\end{itemize}

\subsection{Expanding Closed-Higher-Order Macros}

Dave Prosser's algorithm uses as hide set for a function call
\code{T(Actuals)} the intersection of the hide set \code{HS} of the
function-macro symbol \code{T} and the hide set \code{HS'} of the
closing parenthesis.

In the first-order case, those two hide-sets are always the same: each
syntactic occurrence of \code{T} in the source is followed by the
closing parenthesis, so their hide sets evolve in the exact same way
during expansion.

This is not the case anymore on higher-order examples. Consider our
\code{NIL} example:
\begin{lstlisting}
#define NIL(xxx) xxx
#define G0(arg) NIL(G1)(arg)
#define G1(arg) NIL(arg)
G0(42)  // ~> NIL(G1)(42) ~> G1(42) ~> NIL(42) ~> 42
\end{lstlisting}
In the \code{NIL} reduction sequence, the third term is
\code{G1(42)}; at this point the hide set of \code{G1} contains
\code{G0} and \code{NIL}, but the hide set of the closing parenthesis
contains only \code{G0}. So the function call is expanded in the
intersected hide set \code{G0}, and the occurrence of \code{NIL} that
it creates can be reduced. (C programmers were complaining to the
C standard bodies about implementations that would only consider the
hide set of the function token, and thus refuse to reduce
\code{NIL(42)} further.)

Note that \code{G1} has the hide set \code{G0, NIL} because this
occurrence was created by the expansion of \code{G0}, and then taken
as parameter and returned by the macro \code{NIL}. With our own
algorithm, we would only have \code{G0} in the trace at this point,
as we do not change the trace of parameters when expanding a function
application. In other words, the \code{NIL} comes from the call to
\code{hsadd} on both the function body and the parameters, and the
hide-set intersection can be presented as compensating this extension
of the parameter hide sets -- which we have shown is unnecessary in
the first-order case.

\end{version} 

\end{document}